\documentclass[aps,superscriptaddress,onecolumn,10pt,11pt,tightenlines,prx]{revtex4-2}

\makeatletter
\renewcommand\paragraph[1]{%
  \par\vspace{1ex}%
  \noindent\textbf{#1}\quad
}
\makeatother

\usepackage[margin=1in]{geometry}

\usepackage{setspace}
\usepackage{amsmath,mathtools,amsthm,amssymb}
\usepackage{tabularx}
\usepackage{tabularray}
\usepackage{graphicx,amsmath,amssymb,amscd,amsthm}
\usepackage[utf8]{inputenc}
\usepackage[T1]{fontenc}
\usepackage{lmodern}
\usepackage{xcolor}

\usepackage{enumitem}
\usepackage[pagebackref,colorlinks,citecolor=teal,linkcolor=magenta,bookmarks=true,unicode]{hyperref}

\hypersetup{
pdftitle={The abelian state hidden subgroup problem},
pdfauthor={Marcel Hinsche, 
Jens Eisert, Jose Carrasco},
pdfkeywords={quantum information, abelian, hidden subgroup problem, symmetries, stabilizer states}
}

\usepackage[nameinlink,capitalize]{cleveref}
\crefname{equation}{Eq.}{Eqs.}
\Crefname{equation}{Eq.}{Eqs.}

\newcommand{\wG}{\widehat{G}}

\newcommand{\ep}{\epsilon}

\newcommand{\la}{\lambda}
\newcommand{\om}{\omega}

\usepackage{bbm}
\newcommand{\CC}{{\mathbb C}}
\newcommand{\NN}{{\mathbb N}}

\newcommand{\ZZ}{{\mathbb Z}}

\newcommand{\cH}{{\mathcal H}}

\newcommand{\ket}[1]{|#1\rangle}
\newcommand{\bra}[1]{\langle#1|} 

\newcommand{\tr}{\operatorname{tr}}

\newcommand{\e}{\mathrm e}

\newtheorem{thm}{Theorem}
\newtheorem{defi}{Definition}
\newtheorem{fact}{Fact}
\newtheorem{prop}{Proposition}
\newtheorem{lemma}{Lemma}

\theoremstyle{definition}

\begin{document}

\title{The abelian state hidden subgroup problem: 
\\ Learning stabilizer groups and beyond}

\author{Marcel Hinsche}\email{m.hinsche@fu-berlin.de}

\affiliation{Dahlem Center for Complex Quantum Systems, Freie Universit\"{a}t Berlin, 14195 Berlin, Germany}

\author{Jens Eisert}
\affiliation{Dahlem Center for Complex Quantum Systems, Freie Universit\"{a}t Berlin, 14195 Berlin, Germany}

\author{Jose Carrasco}\email{jose.carrasco@fu-berlin.de}

\affiliation{Dahlem Center for Complex Quantum Systems, Freie Universit\"{a}t Berlin, 14195 Berlin, Germany}

\begin{abstract}
Identifying the symmetry properties of quantum states is a central theme in quantum information theory and quantum many-body physics. In this work, we investigate quantum learning problems in which the goal is to identify a hidden symmetry of an unknown quantum state. Building on the recent formulation of the state hidden subgroup problem (StateHSP), we focus on abelian groups and develop an efficient quantum algorithm that learns any hidden symmetry subgroup using a generalized form of Fourier sampling. We showcase the versatility of the approach in three concrete applications: These are learning (i) qubit and qudit stabilizer groups, (ii) cuts along which a state is unentangled, and (iii) hidden translation symmetries. Through these applications, we reveal that well-known quantum learning primitives, such as Bell sampling and Bell difference sampling, are, in fact, special cases of Fourier sampling. Our results highlight the broad potential of the StateHSP framework for symmetry-based quantum learning tasks and provide protocols that are easier to implement on near-term quantum devices.
\end{abstract}

\maketitle

\vspace{-0.5cm}
\section{Introduction}

Symmetry lies at the heart of modern physics, shaping everything from the conservation laws that govern classical and quantum mechanics to the rich classification of phases of matter in quantum many-body systems. At its core, symmetry captures what remains invariant under specific transformations---typically described by mathematical groups---providing a unifying framework to study the fundamental properties of complex quantum systems.

In the rapidly evolving field of quantum information science, symmetries continue to play a particularly prominent role. They inform the design and analysis of diverse quantum protocols
offering both conceptual clarity and practical advantages. In quantum simulation, symmetry can dramatically reduce resource overhead and enhance accuracy. Symmetries are likewise crucial for quantum error correction \cite{terhalQuantumErrorCorrection2015a} and mitigation strategies, particularly those that rely on entangled copies of quantum states or exploit invariant subspaces \cite{caiQuantumErrorMitigation2023}, as well as for the robust and scalable benchmarking and certification of quantum devices \cite{eisertQuantumCertificationBenchmarking2020a, klieschTheoryQuantumSystem2021a}. In quantum many-body and condensed matter physics, symmetries are essential for understanding conventional phases and exploring exotic topological and symmetry‑protected order \cite{zengQuantumInformationMeets2019a}.

Furthermore, symmetries are a driving force behind many of the most powerful quantum algorithms. For instance, group-theoretic structures underpin the celebrated \emph{quantum Fourier transform} (QFT), a key component in algorithms such as Shor’s for factoring large integers
\cite{shorAlgorithmsQuantumComputation1994a, childsLectureNotesQuantum2025} and other hidden subgroup problems. By systematically exploiting group symmetries, researchers are able to craft algorithms that outperform their classical counterparts.

At the same time, quantum property testing \cite{montanaroSurveyQuantumProperty2016c} and quantum learning theory \cite{arunachalamGuestColumnSurvey2017b,anshuSurveyComplexityLearning2024} have emerged as central frameworks for understanding the capabilities and limitations of extracting information from quantum systems. The construction of learning and testing algorithms---as well as the proof of their optimality---often relies on the structural constraints introduced by symmetry, which enable the design of more efficient procedures that exploit invariance to reduce complexity.

\subsection{State hidden subgroup problem}

This ubiquity of symmetries in quantum states naturally raises the question: can we \textit{learn} an unknown symmetry group directly from copies of a quantum state? In this vein, 
Ref.\ \cite{boulandStateHiddenSubgroup2025} has recently introduced the \emph{state hidden
subgroup problem }(StateHSP) as the task of identifying a hidden symmetry subgroup of some parent group $G$ leaving a quantum state invariant.

\begin{defi}[State hidden subgroup problem (StateHSP)]
Let $G$ be a finite group with a unitary representation $R:G\to \mathrm{U}(\mathcal{H})$ acting on the Hilbert space $\mathcal{H}$
and let $H\leq G$ be a subgroup of $G$. Assume that you have access
to copies of an unknown quantum state vector $\ket{\psi}\in\mathcal{H}$
that is promised to have the following properties:
\begin{enumerate}
\item $\forall\,h\in H,\quad R\left(h\right)\ket{\psi}=\ket{\psi}$.
\item $\forall g\,\not\in H,\quad\left|\bra{\psi}R(g)\ket{\psi}\right|\leq1-\epsilon$.
\end{enumerate}
The problem is to identify $H$.
\end{defi}

The promise ensures that exact symmetries leave $\ket{\psi}$ invariant while any other group element perturbs it by at least $\epsilon$.
As a concrete application of the StateHSP framework, 
Ref.\ \cite{boulandStateHiddenSubgroup2025} has given an efficient algorithm for the \textit{hidden cut problem}---a generalization of entanglement testing. This problem asks to identify the cuts across which a given state
takes the form of a product state. Their algorithm extends to an efficient algorithm for the StateHSP over arbitrary finite abelian groups using $\mathrm{poly}(\log |G|, 1/\epsilon)$ copies of the state. The result demonstrates the framework's power for quantum learning tasks.

The StateHSP formulation builds a bridge between quantum learning theory and traditional quantum algorithms. It generalizes the standard \emph{hidden subgroup problem}  (HSP)~\cite{childsLectureNotesQuantum2025} by hiding subgroup structure in the invariance properties of an unknown quantum state rather than in a classical oracle. In particular, in the extreme case of $\epsilon=1$, any $g\notin H$ yields an orthogonal state, recovering the familiar coset‑state formulation of the standard HSP. However, unlike the HSP, the StateHSP admits no classical analogue, since its input is inherently quantum.

Despite these parallels, the new lens offered by the StateHSP formulation remains largely unexplored: Can other symmetry-related learning problems be formulated as instances of the StateHSP? What is the optimal way to solve the StateHSP and how does it depend on the parameter $\epsilon$?

\subsection{Main results}
\begin{figure*}[t]
    \includegraphics[width=.98\textwidth]{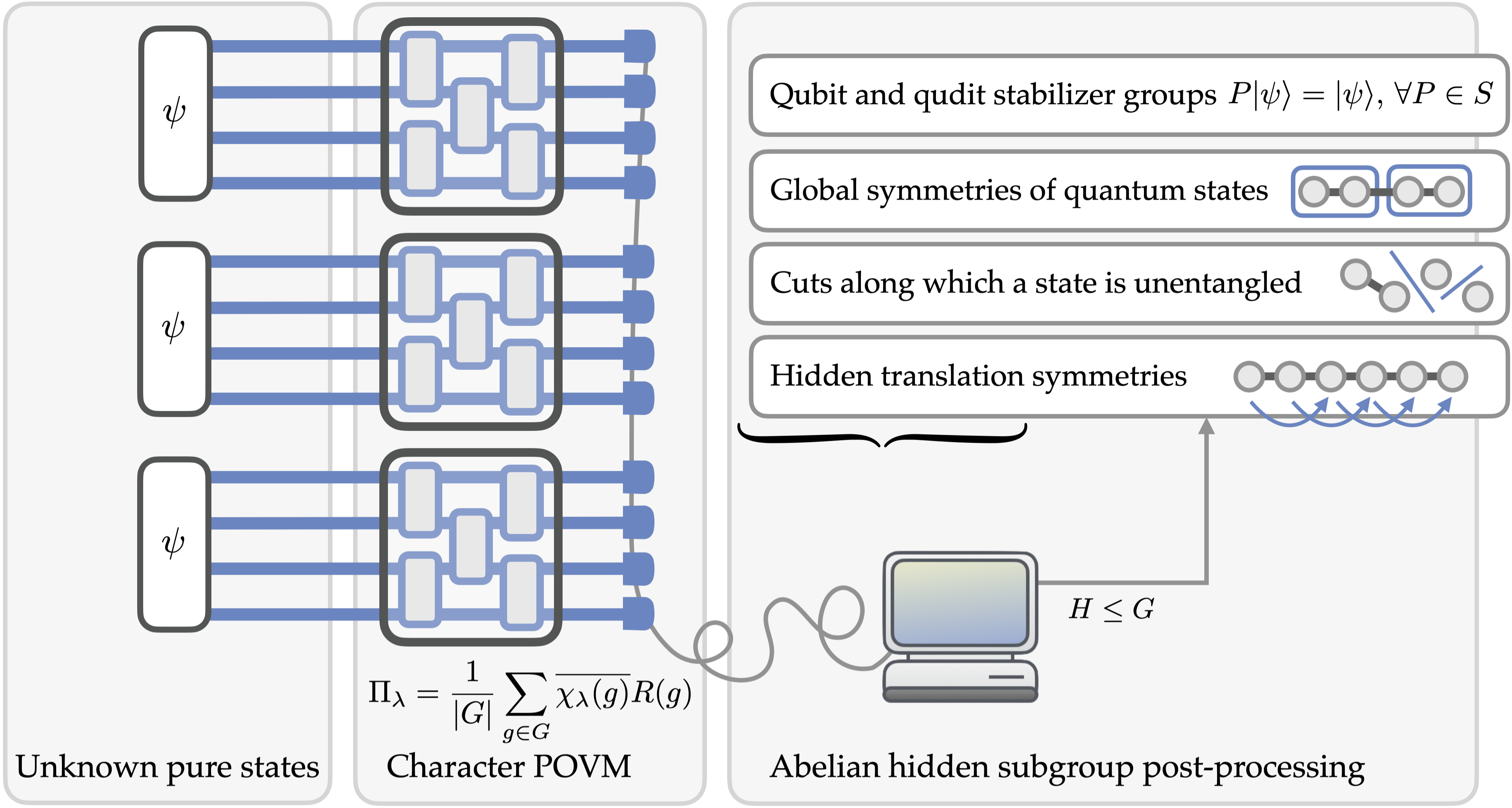}
     \caption{Overview of our work on the abelian StateHSP. We perform Fourier sampling by measuring the character POVM on $m=O(\log |G| /\epsilon)$ copies of the unknown input state vector $\psi=\ket{\psi}\bra{\psi}$, followed by abelian hidden subgroup problem post-processing to recover the hidden subgroup $H\leq G$. The character POVM is efficiently implementable and corresponds to well-known quantum learning measurement routines.}
\label{fig:overview}
\end{figure*}
In this work, we focus on the \textit{abelian} StateHSP and make three main technical and conceptual contributions (see \cref{fig:overview}):
\begin{enumerate}
\item We show that \textit{Fourier sampling}, a standard quantum algorithmic technique for the HSP, enables an efficient solution to any abelian StateHSP using only $O(\log|G|/\epsilon)$ copies. Our analysis is straightforward and significantly simpler than that in 
Ref.\ \cite{boulandStateHiddenSubgroup2025}. As such, it is expected to be more amenable to experimental implementation.

\item We demonstrate applications of the abelian StateHSP framework to new and existing problems—including learning stabilizer groups for qubits and qudits, the hidden cut problem, and learning translational symmetries.
\item We show that for certain groups and representations, Fourier sampling can be implemented without generalized phase estimation~\cite{harrowApplicationsCoherentClassical2005}. Besides enabling more efficient implementations -- and hence contributing to the quest of finding simpler quantum algorithms that are more
feasible to be implemented \cite{MindTheGaps}
-- this observation reveals that established measurement routines in the quantum learning literature, such as \textit{Bell sampling}~\cite{montanaroLearningStabilizerStates2017,huangInformationTheoreticBoundsQuantum2021c,hangleiterBellSamplingQuantum2024} and \textit{Bell difference sampling}~\cite{grossSchurWeylDuality2021,grewalEfficientLearningQuantum2024,chenStabilizerBootstrappingRecipe2025}, are in fact special cases of Fourier sampling.
\end{enumerate}
We now elaborate on each of these contributions in more detail.

\paragraph{Solving the abelian StateHSP via Fourier sampling.}
The key technique for solving the standard abelian HSP is \textit{Fourier sampling} \cite{childsWeakFourierSchurSampling2007,childsQuantumAlgorithmsAlgebraic2010,childsLectureNotesQuantum2025}, which uses the regular representation of $G$ to extract information about the hidden subgroup.
In the StateHSP setting, which involves an arbitrary representation $R$ of $G$, a natural generalization of Fourier sampling can be applied, as observed in Ref.\ \cite{boulandStateHiddenSubgroup2025}. Specifically, Fourier sampling with respect to $R$ corresponds to measuring the positive operator valued measure (POVM) that we call \textit{character POVM}, whose outcomes $\lambda$ correspond to characters $\chi_\la$ of the abelian group $G$. Measuring
\begin{equation}\label{eq:POVM}
\Pi_{\lambda}=\frac{1}{\left|G\right|}\sum_{g\in G}\overline{\chi_\la(g)}R(g) \,
\end{equation}
on a copy of the unknown state vector $\ket{\psi}$ yields the output distribution $q_\psi (\la) =\tr(\Pi_{\la}\ket{\psi}\bra{\psi})$,
which is supported only on $H^\perp$, the dual of the hidden subgroup $H$. Hence, repeated sampling from $q_\psi$ gradually reveals $H$.

The key technical question---also studied in 
Ref.\ \cite{boulandStateHiddenSubgroup2025} for the hidden cut problem---is how the parameter $\epsilon$, absent in the standard HSP, affects sample complexity.
For the hidden cut problem, \cite{boulandStateHiddenSubgroup2025} relied on measuring multiple copies of $\ket{\psi}$ simultaneously to amplify the orthogonality condition $\left|\bra{\psi}R(g)\ket{\psi}\right|\leq 1-\epsilon$ to $\left|\bra{\psi}^{\otimes t}R(g)^{\otimes t}\ket{\psi}^{\otimes t}\right|\leq (1-\epsilon)^t$ for $g \notin H$, where $t$ is the number of coherently accessed copies. They combined this idea with an adaptive variant of Fourier sampling and a technically involved analysis.

In contrast, we show that such amplification and adaptivity are unnecessary. For any abelian StateHSP, non-adaptively sampling $\lambda \sim q_\psi$ a total of $O(\log|G|/\epsilon)$ times and applying standard abelian HSP post-processing (see
Ref.~\cite{childsLectureNotesQuantum2025}), one can recover $H$ efficiently in polynomial time.
This is proved via a straightforward anti-concentration argument.

\begin{thm}[Efficient algorithm for the abelian StateHSP]\label{thm:informal-sample-complexity}
For any finite abelian group $G$ and $\epsilon>0$, there is a polynomial‑time quantum algorithm that, with $O(\log|G|/\epsilon)$ copies of the unknown state vector $\ket\psi$, identifies the hidden subgroup $H\leq G$.
\end{thm} 

\paragraph{Implementing Fourier sampling.}
The standard technique for implementing Fourier sampling is \emph{generalized phase estimation} \cite[Chapter 8]{harrowApplicationsCoherentClassical2005}, which uses an auxiliary register  and controlled application of $R(g)$ together with the quantum group Fourier transform (see \cref{ssec:quantum-fourier-sampling}). While broadly applicable to all finite abelian
groups, this technique requires additional qubits and complex controlled operations, posing practical challenges for near-term devices.

An alternative strategy, conceptually simpler but potentially less general, is to directly measure in a joint eigenbasis of the representation $\{R(g)\}_{g \in G}$. Whether this approach is feasible depends on both the group $G$ and the representation $R$, and in general, it is unclear how to implement such a measurement. Nevertheless, we show that in several applications this approach is natural and easily realized. Crucially, this viewpoint reveals that quantum learning primitives like \textit{Bell sampling}~\cite{montanaroLearningStabilizerStates2017,huangInformationTheoreticBoundsQuantum2021c,hangleiterBellSamplingQuantum2024}, \textit{Bell difference sampling}~\cite{grossSchurWeylDuality2021,grewalEfficientLearningQuantum2024,chenStabilizerBootstrappingRecipe2025}, and \textit{computational difference sampling}~\cite{grewalEfficientLearningQuantum2024,hinscheSingleCopyStabilizerTesting2025} are, in fact, special cases of Fourier sampling, providing a unifying framework and enabling more efficient implementations.

\paragraph{Learning stabilizer groups.}
Learning stabilizer states is a well-studied problem in the qubit
case: The first algorithm to identify an unknown stabilizer state
has been given by Aaronson and Gottesman in Ref.\ \cite{aaronsonIdentifyingStabilizerStates2008}
featuring a sample complexity of $O\left(n^{2}\right)$. Following up on this work,
Montanaro\ \cite{montanaroLearningStabilizerStates2017} has proposed
an algorithm based on Bell sampling that achieves the optimal $O\left(n\right)$
sample complexity. Further studies \cite{grewalEfficientLearningQuantum2024,hangleiterBellSamplingQuantum2024,chiaEfficientLearning$t$doped2024} 
have demonstrated how to learn the stabilizer group of an unknown input
state that is not necessarily a stabilizer state but possesses a
non-trivial stabilizer group. The key measurement routine used in this body of work is that of \emph{Bell difference sampling}, first introduced in Ref.\  \cite{grossSchurWeylDuality2021}.
More recent work \cite{grewalImprovedStabilizerEstimation2024d,chenStabilizerBootstrappingRecipe2025}
has applied Bell difference sampling to \emph{agnostic learning} stabilizer
states, 
where the goal is to output the closest stabilizer state without
any promise on the unknown state.

In contrast, stabilizer learning for qudits (with local dimension
$d>2$) remains underexplored. The sole work addressing qudits explicitly
\cite{allcockBellSamplingStabilizer2024} shows that the key subroutine
for qubits---Bell difference sampling---fails to provide useful information in the qudit 
case.

In this work, we demonstrate that learning stabilizer groups can be formulated as an instance of the abelian StateHSP framework. The key trick is to focus on Pauli group up to phases which is isomorphic to $G =\mathbb{Z}_{d}^{2n}$ where $d$ is the local dimension of the qudits.  Formally, we define the following problem: Let $\mathcal{P}_{n}$ denote the generalized $n$-qudit Pauli group. Then 
a stabilizer group $S$ is an
abelian subgroup of $\mathcal{P}_{n}$ that does not contain non-trivial multiples of the identity operator. Denote by $\mathrm{P}_n:=\mathcal{P}_{n} / \langle \omega I\rangle $ and $\mathrm{S} := S/ \langle \omega I \rangle$ the corresponding \textit{phaseless} or \textit{projective} Pauli group and its induced stabilizer subgroup, where $\omega=e^{2\pi i/d}$.

\begin{defi}[Hidden stabilizer group problem]\label{def:hidden-stabilizer-group}
 Let $\ket{\psi}$ be a state 
 vector on $n$ qudits and suppose that there is a non-trivial stabilizer group $S\subseteq \mathcal{P}_n$ such that
\begin{equation}
P\ket{\psi} = \ket \psi, \quad \forall\, P \in S\,,
\end{equation}
while for all $P \in \mathrm{P}_{n} \setminus \mathrm{S}$, $|\bra{\psi} P \ket{\psi}| \leq 1-\epsilon$.
The hidden stabilizer group problem asks to identify $S$ given access
to copies of $\ket{\psi}$.
\end{defi}

By applying the Fourier sampling approach, we provide a unified algorithm
for this problem for qubits and qudits. Moreover, we show that in the qubit case, Fourier sampling corresponds exactly to the familiar Bell difference sampling routine \cite{montanaroLearningStabilizerStates2017,grewalEfficientLearningQuantum2024} on four copies of $\ket \psi$, whereas in the qudit case, it gives rise to a novel measurement routine.
\begin{thm}[Unified algorithm for learning stabilizer 
groups---Informal version of \Cref{thm:learning-stabilizer-groups}]
There is an efficient non-adaptive quantum algorithm for the 
hidden
stabilizer group problem on $n$ qudits which uses $O\left(n\right/\epsilon)$
copies of the state and runs in polynomial time acting 
coherently
on at most $O\left(d\right)$ copies.
\end{thm}

We note that this algorithm can be applied to learn qudit stabilizer
and $t$-doped stabilizer states \cite{leoneLearningTdopedStabilizer2024d} for $t=O\left(\log n\right)$, since, as shown in Refs.\ \cite{grewalEfficientLearningQuantum2024,hangleiterBellSamplingQuantum2024}, learning such states reduces to learning their stabilizer group, which allows to compress the non-stabilizerness in a small part of the state.
In this sense, our algorithm is more versatile than the one presented
in Ref.\ \cite{allcockBellSamplingStabilizer2024} for 
qudits which is limited to
learning non-doped stabilizer states. Moreover, by detecting global stabilizer-like symmetries, our algorithm naturally extends to identifying features associated with \emph{symmetry-protected topological} (SPT) order~\cite{zengQuantumInformationMeets2019a,elseSymmetryProtectionMeasurementbased2012}, offering a connection to the study of quantum phases of matter in condensed matter physics.

In addition, we note that 
Refs.\ 
\cite{grewalImprovedStabilizerEstimation2024d,chenStabilizerBootstrappingRecipe2025} have recently constructed algorithms for
agnostic learning of stabilizer states in the qubit case. These algorithms crucially rely on Bell difference and, in particular, exploit certain concentration properties of the Bell difference sampling distribution. We observe that the distribution $q_\psi$ arising in our qudit setting exhibits analogous concentration properties. We thus expect that the agnostic stabilizer learning algorithms from
Refs.\ 
\cite{grewalImprovedStabilizerEstimation2024d,chenStabilizerBootstrappingRecipe2025} can be generalized to the qudit case by suitably replacing Bell difference sampling with our qudit measurement routine.

Lastly, we point out that the same algorithm can also be applied in a scenario where one is given many different inputs states sharing a common hidden stabilizer group instead of many copies of a single state. Hence, one can learn a stabilizer code from access to different states in the code space.

\paragraph{Improvement for the hidden cut problem via Bell sampling.}
We revisit the hidden cut problem introduced in Ref.\ \cite{boulandStateHiddenSubgroup2025} as another instance of the StateHSP framework. We directly focus on the general version of the problem involving potentially many cuts, called the hidden
many-cut problem, defined as follows:
\begin{defi}[Hidden many-cut problem]
 Let $\ket{\psi}$ be a state vector on $n$ qubits. Suppose that $\ket{\psi}$
is a product of $m\geq1$ factor state vectors
\begin{equation}
\ket{\psi}=\ket{\phi_{1}}_{C_{1}}\otimes\cdots\otimes\ket{\phi_{m}}_{C_{m}}
\end{equation}
for some partition $C_{1}\sqcup\cdots\sqcup C_{m}=\left[n\right]$
such that each factor state vector $\ket{\phi_{k}}_{C_{k}}$ is at least
$\tilde\epsilon$-far from any separable state on $\left|C_{k}\right|$
qubits. The hidden many-cut problem asks to identify the set partition
$C_{1}\sqcup\cdots\sqcup C_{m}$, given copies of $\ket{\psi}$.
\end{defi}

We emphasize that the parameter $\tilde{\epsilon}$ used in the hidden many-cut problem does not coincide with the $\epsilon$ in \cref{def:stateHSP}. Rather, the two are quantitatively related, with $\tilde{\epsilon}^2$ scaling on the order fo $\epsilon$ (see also \cref{prop:non-cuts}). We adopt the notation $\tilde{\epsilon}$ here to remain consistent with the formulation and complexity bounds in Ref.\ \cite{boulandStateHiddenSubgroup2025}, where the resource requirements are expressed in terms of this parameter.
In Ref.\ \cite{boulandStateHiddenSubgroup2025}, the authors have given an algorithm based on generalized phase estimation requiring controlled-SWAP operations. In particular, their algorithm solves this problem using $O\left(n/{\tilde\epsilon}^{2}\right)$
many samples obtained by measuring the output of circuits of depth $O\left(n^{2}\right)+O\left(\log {\tilde\epsilon}^{-1}\right)$ acting coherently on $O\left(1/{\tilde\epsilon}^{2}\right)$
copies and $O\left(n\right)$ auxiliary qubits at a time. Here, we improve upon this significantly by showing that the problem can be solved by a direct implementation of the character POVM as hinted at earlier. Interestingly, this direct implementation corresponds to another well-known quantum learning subroutine, namely \textit{Bell sampling}~\cite{montanaroLearningStabilizerStates2017}.

\begin{thm}[Improved hidden-cut algorithm---Informal version of \Cref{thm:hidden-cut-thm}]
There is an efficient non-adaptive quantum algorithm for the hidden
many-cut problem on $n$ qubits which uses $O\left(n/{\tilde\epsilon}^{2}\right)$
copies of the state and runs in polynomial time, requiring circuits
of constant depths acting on two copies at a time using no additional
auxiliary qubits.
\end{thm}

This improves upon Ref.\ \cite{boulandStateHiddenSubgroup2025} in all
requirements, namely 
the circuit depth, the number of coherently accessed copies
as well as removing the need for auxiliary registers while keeping the sample and time complexity the same.

\paragraph{Identifying translational symmetries.} As a novel application of the StateHSP framework, we formulate the problem of identifying a hidden translational symmetry of a state. Here, for simplicity, we focus on qubits on a ring in one spatial dimension, although the framework immediately generalizes to translations in higher dimensions, e.g.,  qubits on a two-dimensional cubic lattice.

\begin{defi}[Hidden translation problem]
 Let $\ket{\psi}$ be a state vector on $n$ qubits on a ring. Suppose that
$\ket{\psi}$ is translation-invariant under some subgroup of translations of the qubits. The hidden translation symmetry problem
asks to identify this subgroup.
\end{defi}

\noindent Since translations form a representation of the abelian cyclic group $\mathbb{Z}_n$, this problem is an instance of the abelian StateHSP and by virtue of our \cref{thm:informal-sample-complexity} can be solved efficiently by measuring with the corresponding character POVM.

\begin{thm}[Efficient algorithm for the hidden translation problem---Informal version of \Cref{thm:hidden-translation}]
There is an efficient non-adaptive quantum algorithm for the hidden
translation problem on $n$ qubits which uses $O\left(\log n\right)$ copies
of the state and runs in polynomial time acting on a single copy of $\ket{\psi}$ at a time.
\end{thm}

\subsection{Technical overview}
\paragraph{Dependence of the sample complexity on $\epsilon$.}
To establish our main result (\cref{thm:informal-sample-complexity}), we combine standard facts from the representation theory of finite abelian groups with a straightforward anti-concentration argument. We begin with stating the output distribution of Fourier sampling
\begin{equation}
    q_\psi(\la)= \frac{1}{\left|G\right|}\sum_{g\in G}\overline{\chi_\la(g)} \bra{\psi}R(g) \ket{\psi}.
\end{equation}
By character orthogonality, $q_\psi$ is supported only on $H^\perp$, the dual of the hidden subgroup $H$. Crucially, we show that $q_\psi$ anti-concentrates over all proper subgroups of $H^\perp$. Specifically, for any proper subgroup $K^\perp < H^\perp$, we have
\begin{equation*}
    q_\psi \big(K^\perp \big) \leq 1- \Omega(\epsilon).
\end{equation*}
This anti-concentration ensures that after $O(\log |G| / \epsilon)$ samples from $q_\psi$, one obtains a complete generating set for $H^\perp$ with high probability, and thus efficiently recovers $H$.

\paragraph{Learning abelianization.}
In the context of learning stabilizer groups, we adopt the known strategy of first learning the stabilizer group up to phases~\cite{montanaroLearningStabilizerStates2017,grewalEfficientLearningQuantum2024,hangleiterBellSamplingQuantum2024}. Within our framework, this can be interpreted as learning the \emph{abelianization} of the group. More generally, even for non-abelian groups $G$, one can consider their abelianization $G/[G,G]$, which is always abelian. For example, the Pauli group $\mathcal{P}_n$ abelianizes to $\mathbb{Z}_2^{2n}$, corresponding to Pauli operators modulo phases. This perspective suggests a broader principle: viewing learning tasks through the lens of abelianizations may help uncover new applications of the abelian StateHSP and simplify existing ones.

\paragraph{Known measurement routines are instances of Fourier sampling.}
A separate conceptual contribution of this work is to show that important quantum learning routines such as Bell sampling and Bell difference sampling are special cases of Fourier sampling for appropriate choices of 
group $G$ and representation $R$. In these cases, we show that the character POVM in \cref{eq:povm} associated with Fourier sampling exactly matches the measurement implemented by these routines. Notably, these choices of $G$
and $R$ arise naturally from the structure of the corresponding learning problems.

\subsection{Related work}
Most closely related works have already been discussed in the context of our applications. Beyond these, we note recent works on quantum algorithms for \emph{testing} symmetries rather than learning them~\cite{labordeTestingSymmetryQuantum2023,rethinasamyQuantumComputationalComplexity2025}. These works employ similar representation-theoretic tools. They also clarify distinct notions of symmetry for quantum states in the mixed-state setting that we pick up on in \cref{ssec:mixed-input-states}.

Additionally, Ref.~\cite{bastidasUnificationFiniteSymmetries2025} develops a general circuit framework for implementing character projection, targeting mostly more involved group actions than those considered in this work. Lastly, Ref.~\cite{wakehamInferenceInterferenceInvariance2024} explores the hidden subgroup problem and Fourier sampling approach in the context of designing heuristic quantum machine learning algorithms.

\section{Abelian state hidden subgroup problem}
In this section, we present the Fourier sampling based approach to the abelian StateHSP. This approach is applicable to identifying hidden symmetries associated to subgroups of any finite abelian group $G$. In \cref{ssec:rep-theory-of-abelian-groups}, we give the necessary background on the representation theory of such groups. In \cref{sec:our-approach-to-abelian-state-hsp}, we outline the algorithm, discuss its key properties and prove our first main result, the upper bound of $O(\log|G|/\epsilon)$ on its sample complexity.  In \cref{ssec:quantum-fourier-sampling}, we discuss how to implement the algorithm as efficiently as possible. Lastly, in \cref{ssec:mixed-input-states}, we comment on extending the StateHSP framework to mixed state inputs.

\subsection{Basics of representation theory of abelian groups}\label{ssec:rep-theory-of-abelian-groups}

Let $G$ be a finite abelian group, then a \textit{character} of $G$ is a group homomorphism $\chi:G\to \mathbb{S}^1$ into the circle group $\mathbb{S}^1$. The set of all characters of $G$, denoted $\wG$, forms the dual group of $G$ and is itself abelian. We have that
$G\cong\wG$ and, therefore, the same set of labels can be used for the group and its dual.
Any subgroup $H\leq G$ defines a dual subgroup $H^\perp\leq\wG$, referred to as the {\em annihilator} of $H$, given by
\begin{equation}
H^\perp=\{\chi_\la\in\wG:\chi_\la(h)=1\,,\forall h\in H\} 
\end{equation}
with $|H|\cdot|H^\perp|=|G|$. Equivalently, $H$ is uniquely determined by $H^\perp$ via
\begin{equation}
H=\{h\in G:\chi_\la(h)=1\,,\forall\,\chi_\la\in H^\perp\} 
\end{equation}
and, therefore, if $\{\chi_{\la_1},\ldots,\chi_{\la_k}\}$ is a generating set for $H^\perp$, then
\begin{align}
H &= \bigcap_{i=1}^{k}\ker(\chi_{\lambda_i}), \\
\text{where} \quad 
\ker(\chi_{\lambda}) &= \{g \in G \,:\, \chi_{\lambda}(g) = 1\}.
\end{align}
Further, the characters satisfy orthogonality relations. Here, we only state a special case of these relations that is directly relevant to this work.

\begin{fact}[Character sum]\label{fact:character-sum}
Let $H\leq G$ be a subgroup of $G$ and $H^\perp\leq\wG$ the corresponding annihilator 
subgroup in the dual $\wG$, then
\begin{equation}
\sum_{\la\in H^\perp}\chi_\la(g)=\begin{cases}
|H^\perp|\,,&g\in H,\\
0\,,&g\notin H.
\end{cases}
\end{equation}
\end{fact}

The characters are also central to the representation theory of finite abelian groups. In particular, since every \emph{irreducible representation} of a finite abelian group is one-dimensional, each character corresponds to such a representation. Formally, let $R:G\to \mathrm{U}(\mathcal{H})$ be a unitary representation of the finite group $G$ on the Hilbert space $\mathcal{H}$. Then $\cH$ decomposes into irreducible representations of $G$ as
\begin{equation}\label{eq:Hdecomp}
    \cH\cong\bigoplus_\lambda V_\lambda^{\oplus n_\la}\cong\bigoplus_\la V_\la\otimes\CC^{n_\la} \,,
\end{equation}
where $n_\la$ is the multiplicity of the irreducible representation $V_\lambda$ and $V_\lambda \otimes\CC^{n_\la}$ is called the \emph{$\lambda$-isotypic component} 
of $\mathcal{H}$.
In the abelian case, each $V_\lambda$ is one-dimensional and corresponds to the character $\chi_\lambda$. An orthonormal basis for the Hilbert space can be labeled as
\begin{equation}\label{eq:basis}
    \ket{\la,\nu_\la}:=\ket\la\otimes\ket{\nu_\la}\,,
\end{equation}
where $\la$ runs over irreducible representations and $\nu_\la\in\{1,2,\ldots,n_\la\}$. This basis is also a common eigenbasis for the mutually commuting operators $\{R(g)\}_{g\in G}$ with eigenvalues given by the characters,
\begin{equation}
    R(g) \ket{\la,\nu_\la} = \chi_{\la}(g) \ket{\la,\nu_\la}\,.
\end{equation}
We focus on projectors onto $\lambda$-isotypic components, which in the abelian case are given by,
\begin{equation}\label{eq:isotypic-proj}
\Pi_\la=\sum_{\nu_\la=1}^{n_\la}\ket{\la,\nu_\la}\bra{\la,\nu_\la}\,,
\end{equation}
where $\ket{\la,\nu_\la}$ is the orthonormal basis described in \cref{eq:basis}. In general, these projectors can also be expressed as follows.
\begin{fact}[Character projection formula]\label{fact:character-projection}
Let $G$ be a finite group, $R$ a representation on $\cH$ and let $\chi_\la$ denote the irreducible  characters. Then $\cH$ decomposes as in \cref{eq:Hdecomp} and
\begin{equation}
\Pi_\la=\frac{\dim V_\la}{|G|}\sum_{g\in G}\overline{\chi_\la(g)}R(g)\,.
 \end{equation}
\end{fact}

We finish this section with a general statement about probability distributions $p$ defined on finite abelian groups $G$. Given a set of group elements $\{g_1, \dots, g_m \}$, we denote the 
 subgroup generated by them by $\langle g_1,\dots,g_m \rangle$. Further, we denote the probability mass on a subgroup $K\leq G$ by $p(K)=\sum_{g\in K}p(g)$.

\begin{lemma}[Probability mass on generated subgroup]\label{lem:probability-mass-of-span}
Let $\epsilon, \delta \in (0,1)$ and let $p$ be a distribution over a finite abelian group $G$. Suppose $g_1,\dots,g_m$ are i.i.d.\  samples from $p$ with 
\begin{equation}
m\geq \frac{2(\log |G|+\log(1/\delta))}{\epsilon}, 
\end{equation}
then with probability at least $1-\delta$,
\begin{equation}
    p(\langle g_1,\dots,g_m  \rangle)\geq 1-\epsilon\,.
\end{equation}
\end{lemma}
Similar statements have been used in Refs.\ \cite{grewalEfficientLearningQuantum2024,chenStabilizerBootstrappingRecipe2025} for distributions over $\mathbb{F}_2^n$ in the context of stabilizer learning. Intuitively, with every new independent sample, the chance of making progress in growing the generated subgroup is at least $\epsilon$ but the number of generators is bounded by $\log |G|$. The result then follows from a Chernoff bound.

\subsection{Solving the abelian StateHSP via Fourier sampling}\label{sec:our-approach-to-abelian-state-hsp}
The state hidden subgroup problem as introduced in Ref.\ \cite{boulandStateHiddenSubgroup2025} is the following problem:
\begin{defi}[State hidden subgroup problem (StateHSP) \cite{boulandStateHiddenSubgroup2025}]\label{def:stateHSP}
Let $G$ be a finite group with a unitary representation $R:G\to \mathrm{U}(\cH)$ acting on the Hilbert space $\cH$ and let $H\leq G$ be a subgroup. Assume that you have access
to copies of an unknown quantum state vector $\ket\psi\in\cH$ with the promise that
\begin{enumerate}
    \item  $R(h)\ket\psi=\ket\psi,\; \forall h\in H$, and
    \item $|\bra\psi R(g)\ket\psi|<1-\ep$, whenever $g\notin H$, where $\epsilon >0$.
\end{enumerate}
The problem is to identify $H$.
\end{defi}

Below we will describe our Fourier sampling-based approach to solving the StateHSP, in the case where $G$ is abelian. Fourier sampling corresponds to measuring $\ket\psi$ repeatedly using the \textit{character POVM}
\begin{equation}
\label{eq:povm}
\Pi_{\lambda}=\frac{1}{\left|G\right|}\sum_{g\in G}\overline{\chi_\la(g)}R(g)\,,
\end{equation}
which is precisely the projector onto the $\la$-isotypic component of the representation $R$ by~\Cref{fact:character-projection}.

Throughout, we will denote the output distribution from measuring the POVM $\Pi_\lambda$ as
\begin{equation}\label{eq:q_psi}
    q_\psi(\la)=\tr(\Pi_\la\ket\psi\bra\psi) = \frac{1}{\left|G\right|}\sum_{g\in G}\overline{\chi_\la(g)} \bra{\psi}R(g) \ket{\psi} .
\end{equation}
As a consequence of~\Cref{fact:character-sum}, we find that for any subgroup $K\leq G$ the probability mass of $q_\psi$ on the annihilator $K^\perp$ satisfies the following relation.
\begin{lemma}[$q_\psi$-mass on subgroups]\label{lem:q_psi-weight}
Let $K\leq G$ be a subgroup of $G$ and let $K^\perp$ be its corresponding annihilator. Then,
    \begin{equation}
        q_\psi(K^\perp) = \tr(\Pi_K \ket{\psi}\bra{\psi} )  = \frac{1}{|K|} \sum_{k\in K} \bra{\psi} R(k)\ket{\psi}\,, \label{eq:subgroup-mass}
    \end{equation}
    where $\Pi_K= \frac{1}{|K|} \sum_{k\in K} R(k)$.
\end{lemma}
\noindent Note that $\Pi_K= \frac{1}{|K|} \sum_{k\in K} R(k)$ is the projector onto the subspace that is invariant under the $K$-action. So, the $q_\psi$-mass on some $K^\perp$ measures precisely how close $\ket{\psi}$ is to being invariant under $K$.
\begin{proof}
Using~\Cref{fact:character-sum}, we find that
\begin{align}
\sum_{\lambda\in K^{\perp}}\Pi_\lambda& =\frac{1}{\left|G\right|}\sum_{g\in G}\Big(\sum_{\lambda\in K^{\perp}}\overline{\chi_\la(g)}\,\Big) R(g) \,,\\
\nonumber
&=\frac{|K^{\perp}|}{\left|G\right|}\sum_{k\in K} R(k) \,,\\
\nonumber
&= \frac{1}{\left|K\right|}\sum_{k\in K} R(k)\,, \\
\nonumber
&= \Pi_K \,,
\end{align}
where we have used $\big|K^{\perp}\big|\big|K\big| =|G|$. By linearity of the trace, we find,
\begin{equation}
    \sum_{\lambda\in K^{\perp}}q_{\psi}(\lambda) = \sum_{\lambda\in K^{\perp}}\tr(\Pi_\la\ket\psi\bra\psi) = \tr(\Pi_K \ket{\psi}\bra{\psi} ).
\end{equation}
\end{proof}

Hence, if $\ket\psi$ is invariant under the $R$-action of the hidden
subgroup $H\leq G$ as assumed in \cref{def:stateHSP}, then
$q_\psi$ is supported only on $H^\perp$, 
meaning
\begin{equation}
q_\psi(H^\perp) =1 \,.
\end{equation}
From an algorithmic viewpoint,
this suggests that repeatedly sampling from $q_\psi$ gradually reveals $H^\perp$. In particular, a sufficiently large number of i.i.d.\  
samples is likely to form a generating set for $H^\perp$. But exactly how many samples are sufficient to ensure this?

The key factor here is the parameter
$\epsilon$ from \cref{def:stateHSP}. Generally, the larger $\epsilon$, the more uniformly $q_\psi$ is distributed over $H^\perp$, increasing the likelihood that new samples remain independent of those already obtained.  
In what follows, we present a straightforward analysis that quantifies this relationship, showing that $O(\log(|H^\perp|/\epsilon))$ samples suffice to form a generating set for $H^\perp$. This analysis is based on demonstrating that $q_\psi$ is anti-concentrated on the proper subgroups of $H^\perp$.
 
\begin{lemma}[Anti-concentration of $q_\psi$]\label{lem:anti-concentration}
Let $G,H,R$ and $\ket{\psi}$ be defined as in \Cref{def:stateHSP}. In particular, assume that $\ket{\psi}$ is such that
\begin{enumerate}
    \item  $R(h)\ket\psi=\ket\psi,\; \forall h\in H$, and
    \item $|\bra\psi R(g)\ket\psi|<1-\ep$, whenever $g\notin H$.
\end{enumerate}
Let $q_\psi$ be defined as in \cref{eq:q_psi}.
Then, for all proper subgroups $K^{\perp}<H^{\perp}$ of $H^{\perp}$, it holds that,
\begin{equation}
q_\psi(K^\perp)  \leq1-\frac{\epsilon}{2}.
\end{equation}
\end{lemma}

\begin{proof}
First, from \cref{lem:q_psi-weight}, we have
\begin{equation}
    q_\psi(K^\perp) = \frac{1}{|K|}\sum_{k\in K}\bra{\psi} R(k) \ket{\psi}.
\end{equation}
 Note that by assumption $H<K$, in fact, $K$ is strictly larger than $H$. Next, we split the sum over $K$ and use the assumptions on $\ket{\psi}$,
 \begin{align}
 q_\psi(K^\perp) & =\frac{1}{|K|}\left(\sum_{k\in H}1+\sum_{k\in K\setminus H}\bra{\psi} R(k) \ket{\psi}\right)\\
 \nonumber
 & \leq\frac{1}{|K|}\left(\left|H\right|+\left(1-\epsilon\right)\left(\left|K\right|-\left|H\right|\right)\right)\,,\\
 \nonumber
 & \leq1-\epsilon/2 \,,
\end{align}
where in the last step, we have used that $|H| \leq |K|/2$ by Lagrange's theorem.
\end{proof}

Combining this anti-concentration result with \Cref{lem:probability-mass-of-span}, we arrive at the following sample complexity upper bound for the abelian StateHSP.
\begin{thm}[Sample complexity upper bound]\label{thm:sample-complexity}
   There exists an algorithm for solving the abelian StateHSP with sample complexity $O(\log\left|G\right|/\epsilon)$.
\end{thm}
\begin{proof}
Let $\epsilon' > \epsilon/2$ and assume that we have sampled  
\begin{equation}
m\geq \frac{2\log |G|+2\log(1/\delta)}{\epsilon'} 
\end{equation}
times from $q_\psi$ to obtain 
i.i.d.\  samples $ \chi_{\la_1},\ldots,\chi_{\la_m}$. Then, it follows from \cref{lem:probability-mass-of-span} that with probability $1-\delta$, the mass on the group generated by the samples satisfies $q_\psi \left(\langle \chi_{\la_1},\ldots,\chi_{\la_m} \rangle \right)> 1-\epsilon/2$. Now, using \cref{lem:anti-concentration}, we conclude that, in this case $\langle \chi_{\la_1},\ldots,\chi_{\la_m} \rangle = H^\perp$ because all proper subgroups of $H^\perp$ only account for at most $1 -\epsilon/2$ of the total mass of $q_\psi$. 
\end{proof}

The StateHSP as formulated in \cref{def:stateHSP} asks us to find a hidden symmetry in a quantum state vector $\ket{\psi}$ promised to have one.
However, sometimes it might be more natural to adopt a slightly different perspective: rather than assuming the existence of a hidden symmetry, we simply ask whether the state exhibits any non-trivial symmetry at all and if so to find it. In this vein, we point out two basic properties about the output of the algorithm:

\begin{lemma}[Properties of $H$  output by the algorithm]\label{lemma:output-prop}
    Let $G$ be a finite group with a unitary representation $R:G\to\mathrm{U}(\cH)$ on $\cH$. Assume that you have access to copies of an unknown quantum state vector $\ket{\psi}$. Let $\{\chi_{\la_1},\ldots,\chi_{\la_m}\}$ be $m$ independent samples from $q_\psi$ and let $H=\cap_i^m \ker (\chi_{\lambda_i})$ be the output of the algorithm. Then:
    \begin{enumerate}
        \item For all $g\in G$ such that $R(g)\ket\psi=\ket\psi$, $g\in H$.
        \item If $m\geq 2(\log |G|+\log(1/\delta))/\epsilon$, then with probability at least $1-\delta$ all elements $h\in H$ satisfy, $|\bra\psi R(h)\ket\psi|\geq 1-\epsilon$.
    \end{enumerate}
\end{lemma}
The first property ensures that the algorithm’s output always includes the exact symmetry group, while the second implies that any elements $g$ not corresponding to exact symmetries are rapidly suppressed as the number of samples increases.

\begin{proof}
Let $g$ be an arbitrary element of $G$ and let $K:=\langle g \rangle$ be the group generated by $g$. Then, the dual of $K$ is given by
\begin{equation}
    K^\perp=\{\la:\chi_\la(g')=1, \forall g'\in \langle g \rangle \}=\{\la:\chi_\la(g)=1\}\,.
\end{equation}
The probability that $g$ belongs to the subgroup $H$ output by our algorithm is given by
\begin{equation}\label{eq:probq}
    \Pr(g\in H)=\Pr(g\in\cap_i^m\ker(\chi_{\lambda_i}))= \prod_{i=1}^m \Pr(g\in\ker (\chi_{\lambda_i}))=\big(q_\psi(K^\perp)\big)^m\,,
\end{equation}
where we have used that the $\lambda_i$ have been sampled independently.
We will now find a bound on this probability in terms of $|\bra{\psi} R(g) \ket{\psi}|$. 
By the orthogonality relation of the characters, we can invert \cref{eq:q_psi} to obtain
\begin{equation}\label{eq:split-up-contributions}
    \bra\psi R(g)\ket\psi=\sum_\la\chi_\la(g)q_\psi(\la)=q_\psi(K^\perp)+\sum_{\la\not\in K^\perp}\chi_\la(g)q_\psi(\la)\,,
\end{equation}
where we have used the definition of $K^\perp$ and written $q_\psi(K^\perp)=\sum_{\la\in K^\perp}q_\psi(\la)$.
From \cref{eq:split-up-contributions}, the first claimed property follows: By taking the real part of this equation, we observe that $\bra\psi R(g)\ket\psi=1$ implies that $q(K^\perp)=1$ and so $ \Pr(g\in H)=1$ by \cref{eq:probq}.

To show the second property, we can lower bound the absolute value of the LHS of \cref{eq:split-up-contributions} as
\begin{equation}
    |\bra\psi R(g)\ket\psi|\geq q_\psi(K^\perp)-\Big|\sum_{\la\not\in K^\perp}\chi_\la(g)q_\psi(\la)\Big|\geq2q_\psi(K^\perp)-1\,,
\end{equation}
where we have used that $q_\psi$ is normalized such that
\begin{equation}
    \Big|\sum_{\la\not\in K^\perp}\chi_\la(g)q_\psi(\la)\Big|\le\sum_{\la\not\in K^\perp}q_\psi(\la)=1-q_\psi(K^\perp)\,.
\end{equation}
Thus, we have shown
\begin{equation}\label{eq:ineq}
    q_\psi(K^\perp)\le\frac{1+|\bra\psi R(g)\ket\psi|}2\,,
\end{equation}
and hence
\begin{equation}
    \Pr(g\in H) \leq \left(\frac{1+|\bra\psi R(g)\ket\psi|}2 \right)^m.
\end{equation}
Finally, consider the event that $H$ contains at least a single $g$ such that $|\bra\psi R(g)\ket\psi| < 1-\epsilon$,
\begin{equation}
A=\bigcup_{\{g\in G:\,|\bra\psi R(g)\ket\psi|<1-\ep\}}\{g\in H\}\,.
\end{equation}
Now, using the union bound, the probability of this event is bounded as
\begin{equation}
    \Pr( A ) \le|G|(1-\ep/2)^m\le\exp(\ln|G|-\ep m/2) \leq \delta,
\end{equation}
whenever the sample complexity satisfies $m\geq \frac{2\log |G|+2\log(1/\delta)}{\epsilon}$.
\end{proof}

\subsection{Implementing Fourier sampling}\label{ssec:quantum-fourier-sampling}
While we have proved that Fourier sampling provides an efficient algorithm for solving the abelian StateHSP, we have not yet detailed how to implement it in practice. The canonical technique for this is known as \emph{generalized phase estimation}. Here, for the sake of concreteness, we show that, in the abelian case, generalized phase estimation implements the character POVM in \cref{eq:povm}. 

Let $G$ be a finite abelian group, then the \emph{quantum Fourier transform} (QFT) is defined as the following unitary \cite{childsLectureNotesQuantum2025}
\begin{equation}
    QFT \ket {g} = \frac{1}{\sqrt{|G|}}\sum_{\lambda} \chi_{\lambda}(g) \ket{\lambda}\,.
\end{equation}
To implement the POVM from \cref{eq:povm} via a quantum Fourier
sampling approach, we proceed as follows:
\begin{enumerate}
\item Start with the state vector $\ket 0\otimes\ket{\psi}$.
\item Put the auxiliary register in superposition to obtain $\frac{1}{\sqrt{\left|G\right|}}\sum_{g\in G}\ket g \otimes\ket{\psi}$.
\item Apply the controlled group action $\sum_{g}\ket g\bra g\otimes R\left(g\right)$
to obtain $\frac{1}{\sqrt{\left|G\right|}}\sum_{g\in G}\ket g \otimes R\left(g\right)\ket{\psi}$.
\item Apply the inverse QFT to the auxiliary register to obtain
$\frac{1}{\left|G\right|}\sum_{\lambda}\sum_{g\in G}\overline{\chi_{\lambda}\left(g\right)}\ket{\lambda}\otimes R\left(g\right)\ket{\psi}$.
\item Measure the auxiliary register in the $\ket{\lambda}$-basis.
\end{enumerate}
This results in the following output distribution
\begin{equation}
q_{\psi}\left(\lambda\right)=\frac{1}{\left|G\right|}\sum_{g\in G} \overline{\chi_{\lambda}\left(g\right)}\bra{\psi}R(g)\ket{\psi}\,.
\end{equation}
This is precisely the distribution $q_\psi (\lambda) = \tr(\Pi_{\la} \ket{\psi}\bra{\psi})$ from \cref{eq:q_psi}.

While the generalized phase estimation approach is broadly applicable to all finite abelian groups, it poses practical challenges. In particular, it requires auxiliary qubits and controlled operations involving many qubits, which can be difficult to implement on near-term quantum hardware, especially when locality constraints limit the connectivity.

An alternative strategy, conceptually simpler but potentially less general, is to directly measure in a joint eigenbasis of the representation $\{R(g)\}_{g \in G}$. Since the $R(g)$ commute for abelian $G$, there exists an orthonormal basis ${\ket{\lambda, \nu_\lambda}}$ of joint eigenvectors (see \cref{eq:basis}) satisfying
\begin{equation}
R(g) \ket{\lambda, \nu_\lambda} = \chi_{\lambda}(g) \ket{\lambda, \nu_\lambda}.
\end{equation}
Measuring in this basis and discarding the multiplicity label $\nu_\lambda$ yields a sample from the same distribution $q_\psi(\lambda)$ as in generalized phase estimation. 

Whether this approach is feasible depends on both the structure of the group $G$ and the explicit form of the representation $R$. In general, 
it is unclear how to construct and implement a measurement in the joint eigenbasis. Nevertheless, in the specific applications of the StateHSP framework explored in this paper, we encounter cases where this direct measurement strategy is natural and easily realizable, offering a more efficient route to Fourier sampling in practice. Interestingly, this perspective also allows us to connect Fourier sampling to well-established measurement routines such as Bell sampling and Bell difference sampling in the quantum learning literature. In particular, in \cref{sec:hidden-cut}, we show that the hidden cut problem defined in \cite{boulandStateHiddenSubgroup2025} can be solved via Bell sampling which precisely corresponds to such a direct implementation of Fourier sampling. 

\subsection{Mixed input states}\label{ssec:mixed-input-states}
An extension of the StateHSP is to allow \emph{mixed} input states. Here, we provide a natural definition for this setting:

\begin{defi}[Mixed state hidden subgroup problem (MixedStateHSP)]\label{def:mixed-stateHSP}
Let $G$ be a finite group with a unitary representation $R:G\to\mathrm{U}(\cH)$ acting on the Hilbert space $\cH$ and let $H\leq G$ be a subgroup. Assume that you have access
to copies of an unknown quantum state $\rho\in \mathcal{L}(\cH)$ with the promise that
\begin{enumerate}
     \item  $ R(h) \rho R(h)^{\dagger} = \rho,\; \forall h\in H$, and
    \item $\lVert R(g) \rho R(g)^{\dagger}-\rho \rVert \geq \epsilon$, whenever $g\notin H$, for $\epsilon >0$ and some suitable norm $\lVert\cdot\rVert$ such as the trace norm.
\end{enumerate}
The problem is to identify $H$.
\end{defi}
This formulation reduces to the StateHSP as defined in \cref{def:stateHSP} in case $\rho$ is pure.
To tackle an abelian MixedStateHSP with a mixed input $\rho$, a first thought is to apply the same character-measurement and post-processing strategy as in the pure state case, setting
\begin{equation}
    q_\rho(\lambda ) = \tr(\Pi_\la \rho) \,.
\end{equation}
However, we note that, in general this strategy fails, as the invariance condition $ R(h) \rho R(h)^{\dagger} = \rho$ in \cref{def:mixed-stateHSP} does not guarantee that $q_\rho$ fully concentrates all its mass on the true annihilator $H^\perp$ such that, in general, $q_\rho (H^\perp ) < 1$. This breaks the algorithm, as we have no efficient way of distinguishing samples belonging to $H^\perp$ from those that do not.

On the other hand, our algorithm \emph{does} succeed when we impose a stronger requirement than that stated in \cref{def:mixed-stateHSP}. Namely,
if $\rho$ satisfies
\begin{enumerate}
    \item   $\tr(R(h)\rho)=1,\; \forall h\in H$, and
    \item $|\tr(R(g)\rho)|\leq 1 - \epsilon$,  whenever $g\notin H$, for some $\epsilon >0$,
\end{enumerate}
then the character POVM-based approach will succeed in learning $H$, even with mixed input states.
This condition is strictly stronger than the invariance requirement in \cref{def:mixed-stateHSP}, as it implies---but is not implied by---that condition. Specifically,
\begin{equation}
     \tr(R(h)\rho)=1 \Rightarrow  R(h) \rho R(h)^{\dagger} = \rho\,.
\end{equation}

We note that these inequivalent notions of symmetries of mixed states have been discussed in detail in \cite{labordeTestingSymmetryQuantum2023} in the context of testing symmetries rather than learning them.
Designing an efficient algorithm that works under the weaker, more natural invariance condition of \cref{def:mixed-stateHSP} remains an open problem. Such a result would broaden the applicability of the StateHSP framework to abelian symmetry learning for general quantum channels. In particular, via the \emph{Choi Jamiołkowski isomorphism}, these problems reduce naturally to instances of the StateHSP. However, the limitations with mixed input states discussed here restrict our current method to unitary channels.

\section{Learning stabilizer groups}
In this section, we explain how the problem of learning stabilizer groups can be viewed as an instance of the abelian StateHSP. We start by giving some essential background on Weyl operators in \cref{ssec:weyl-operators}. Then, in \cref{ssec:learning-weyl}, we show how to solve the problem via Fourier sampling in \cref{sec:our-approach-to-abelian-state-hsp}. In \cref{ssec:common-eigenbasis-of-weyls}, we show that for this problem, Fourier sampling can always be implemented by a suitable Clifford circuit. Lastly, in \cref{ssec:learning-code}, we argue that the algorithm can also be applied to learning a stabilizer code when given access to different code states rather than many copies of the same one.

\subsection{Weyl operators and stabilizer groups}
\label{ssec:weyl-operators}
Let $d\geq 2$ be a prime. Consider a single qudit with computational
basis elements $\ket{q}$ where $q\in\mathbb{Z}_{d}$. Define
the unitary shift and clock operators $X$ and $Z$, respectively,
as
\begin{equation}\label{eq:XZ}
\begin{aligned}
X\ket q & =\ket{q+1},\\
Z\ket q & =\omega^{q}\ket q,
\end{aligned}
\end{equation}
for all $q\in\ZZ_d$, where $\omega=\e^{2\pi i/d}$ is the $d$-th root of unity.
There are some slight differences in the algebra of the shift and clock operators between the qubit case ($d$ even) and the qudit case ($d$ odd). To treat both cases simultaneously, we also
introduce $\tau=\e^{i\pi (d^2+1)/d}$. We have that that $\tau^{2}=\omega$ and
\begin{equation} \tau^d =
    \begin{cases}
     1 &d \text{ odd} \,,\\
     - 1 &d \text{ even}\,.
\end{cases}
\end{equation}
Next, we introduce the $n$-qudit Weyl operators as
\begin{equation}
W_{x}=W_{(a,b)}=\tau^{-a\cdot b}\left(Z^{a_{1}}X^{b_{1}}\right)\otimes\cdots\otimes\left(Z^{a_{n}}X^{b_{n}}\right)
\end{equation}
for all $x=\left(a,b\right)\in\ZZ_d^{2n}$. Each Weyl operator is also an element of the generalized Pauli group.

Consider the $\mathbb{Z}_{d}$-valued symplectic form defined in $\ZZ_d^{2n}$ for
$x=(a,b)$ and $y=(a',b')$ as
\begin{equation}
[x,y]=\left[\left(a,b\right),\left(a',b'\right)\right]=a\cdot b'-a'\cdot b\,\mod\,d\,.
\end{equation}
Then, the Weyl operators compose as
\begin{equation}\label{eq:tau}
W_{x}W_{y}=\tau^{[x,y]}W_{x+y}\,,
\end{equation}
and their commutation relations are captured by the following equation,
\begin{equation}\label{eq:omega}
W_{x}W_{y}=\omega^{\left[x,y\right]}W_{y}W_{x}\,.
\end{equation}
For more background on Weyl operators and the symplectic formalism, see Ref.\ \cite{grossSchurWeylDuality2021}.\\

\noindent Tensor products of the shift and clock operators generate the generalized $n$-qudit Pauli group,
\begin{equation}
    \mathcal{P}_n := \langle \tau I, X, Z\rangle^{\otimes{n}}\,.
\end{equation}
A stabilizer group $S$ is an abelian subgroup of the generalized Pauli group $\mathcal{P}_n$ that does not contain a non-trivial multiple of the identity operator $I^{\otimes{n}}$. We say that a state vector $\ket{\psi}$ has a non-trivial stabilizer group if there exists $S \neq \{I^{\otimes{n}}\}$, such that
\begin{equation}
    P\ket \psi = \ket{\psi},\quad \forall P \in S \,.
\end{equation}
For our purposes, it will be useful to consider the Pauli operators only up to phase. To this end, we denote by $   \mathrm{P}_n:=\mathcal{P}_{n} / \langle \omega I\rangle$ the phaseless or \textit{projective} Pauli group which is isomorphic to $\mathbb{F}_d^{2n}$.
For any state vector $\ket \psi$, we also define its corresponding phaseless stabilizer group as follows, adopting the notation introduced in Ref.~\cite{grewalEfficientLearningQuantum2024}:
\begin{defi}[Phaseless stabilizer group]
Let $\ket\psi\in(\CC^{d})^{\otimes n}$ be a state vector on $n$ qudits. We define $\mathrm{Weyl}\left(\ket{\psi}\right)$ as the following subspace of $\mathbb{F}_d^{2n}$, 
\begin{equation}
\mathrm{Weyl}\left(\ket{\psi}\right)=\left\{ x\in\mathbb{F}_{d}^{2n}:W_{x}\ket{\psi}=\omega^{s} \ket\psi\:\text{for some }s\in\mathbb{F}_{d}\right\} \,.
\end{equation} 
\end{defi}
\noindent We note that $\mathrm{Weyl}\left(\ket{\psi}\right)\subset \mathbb{F}_d^{2n}$ is isomorphic to $\mathbb{F}_{d}^{k}$ where $k=\dim\mathrm{Weyl}\left(\ket{\psi}\right)\leq n$. 

\subsection{Learning \texorpdfstring{$\mathrm{Weyl}\left(\ket{\psi}\right)$}{phaseless stabilizer group} as an abelian
StateHSP}\label{ssec:learning-weyl}
We start by formulating the problem of learning stabilizer groups up to phases.
\begin{defi}[Hidden phaseless stabilizer group problem]
\label{def:hidden-phaseless-stabilizer-group}
 Let $\ket{\psi}\in\CC^{d^n}$ be a state vector on $n$ qudits and suppose that it has a non-trivial phaseless stabilizer group $\mathrm{Weyl}(\ket{\psi})$ while for all $x \not\in \mathrm{Weyl}(\ket{\psi})$, $|\bra{\psi} W_x \ket{\psi}| \leq 1-\epsilon$.
The hidden phaseless stabilizer group problem asks to identify $\mathrm{Weyl}(\ket{\psi})$ given access
to copies of $\ket{\psi}$.
\end{defi}
We can identify the phase-full hidden stabilizer group $S$ mentioned in \Cref{def:hidden-stabilizer-group} from its corresponding phaseless version $ \mathrm{Weyl}\left(\ket{\psi}\right)$ by determining the missing phases of the stabilizers from measuring a set of generators of $\mathrm{Weyl}\left(\ket{\psi}\right)$ on $\ket{\psi}$.

\noindent The phaseless version of the hidden stabilizer group problem in \cref{def:hidden-phaseless-stabilizer-group} fits into the StateHSP framework by choosing $G =\mathbb{Z}_{d}^{2n}$. To treat the qubit and the qudit case simultaneously, let $D$ denote the order of $\tau$, i.e.,
\begin{equation} D =
    \begin{cases}
     d &d \text{ odd} \,,\\
      2d &d \text{ even}\,,
\end{cases}
\end{equation}
and choose the representation
\begin{align}
    R: \mathbb{Z}_{d}^{2n} &\to \mathrm{U}\big((\CC^{d^n})^{\otimes{D}} \big), \\
    x &\mapsto R(x) = W_x^{\otimes D} .
\end{align}
It follows from \Cref{eq:tau} that this is a valid representation, however, we emphasize that $R$ is a representation on $(\CC^{d^n})^{\otimes{D}}$, i.e., on $D$ copies of the $n$-qudit Hilbert space $\CC^{d^n}$.
The hidden subgroup is given by
\begin{equation}
    H = \mathrm{Weyl}\left(\ket{\psi}\right) \,,
\end{equation}
and we have for $\ket{\psi}^{\otimes D}$ that
\begin{enumerate}
    \item $W_x^{\otimes D}\ket{\psi}^{\otimes D} = \ket{\psi}^{\otimes D}$, for all $x\in \mathrm{Weyl}(\ket{\psi})$, 
    \item $|\bra{\psi}W_x \ket{\psi}|^D\leq (1-\epsilon)^D$, for all $x\not\in\mathrm{Weyl}(\ket{\psi})$.
\end{enumerate}

Since $\mathbb{Z}_{d}^{2n}$ is abelian, we can tackle this problem via the Fourier sampling approach to the abelian StateHSP outlined in \cref{sec:our-approach-to-abelian-state-hsp}. In particular, we will implement the character POVM corresponding to \Cref{eq:povm}. The characters are labeled by $y=(a,b)\in \mathbb{Z}_{d}^{2n}$ and given by
\begin{equation}
    \chi_{y}(x) = \om^{[x,y]}\,,
\end{equation}
for all $x\in\mathbb{Z}_{d}^{2n}$, and the POVM is given by
\begin{equation}\label{eq:povm-weyl}
    \Pi_y = \frac{1}{d^{2n}}\sum_{x\in \mathbb{Z}_{d}^{2n}} \om^{-[y,x]} W_x^{\otimes D}\,.
\end{equation}
Hence, the output distribution $q_\psi=\tr \big (\Pi_y\; (\ket{\psi}\bra{\psi})^{\otimes D} \big)$ is given by 
\begin{equation}\label{eq:qgoal}
    q_{\psi}(y)= \frac{1}{d^{2n}}\sum_{x\in \mathbb{Z}_{d}^{2n}}\om^{-[y,x]} \bra\psi W_x\ket\psi^{D}\,.
\end{equation}
In the qubit case, where $d=2$ and $D=4$, \cref{eq:povm-weyl} is the Bell difference sampling POVM \cite{grossSchurWeylDuality2021, grewalEfficientLearningQuantum2024}. Hence, in the qubit case, Fourier sampling corresponds exactly to known algorithms for learning stabilizer groups. However, the approach also naturally generalizes to qudits, where such algorithms were not previously known.

The projective measurement $\Pi_y$ from \Cref{eq:povm-weyl} can be implemented by measuring in the common eigenbasis of the $R(x) = W_x^{\otimes D}$ for all $
x\in \mathbb{Z}_{d}^{2n}$. We comment on this eigenbasis in the subsection below.
Here, by virtue of our general result, \cref{thm:sample-complexity}, we arrive at the following theorem.
\begin{thm}[Unified hidden stabilizer group algorithm]\label{thm:learning-stabilizer-groups}
There is an efficient non-adaptive quantum algorithm for the hidden
stabilizer group problem on $n$ qudits which uses $O\left(n \log d \max\{d, 1/\epsilon\} \right)$
copies of the unknown state and runs in polynomial time, 
requiring circuits
of depths $O(d)$ acting coherently on at most $d$ copies at a time using no additional
auxiliary systems. 
\end{thm}
Lastly, we observe that computational difference sampling, another primitive from the stabilizer learning literature \cite{grewalEfficientLearningQuantum2024, hinscheSingleCopyStabilizerTesting2025}, also can be viewed as Fourier sampling. It involves measuring two copies of an $n$-qubit state vector $\ket{\psi}$ in the computational basis and combining the outcomes via bitwise XOR. This approach enables learning $Z$-like stabilizers of $\ket{\psi}$ (up to phase) using only single-copy measurements. This measurement routine corresponds to Fourier sampling with respect to $G=\mathbb{Z}_2^n$ and the representation $R:a \mapsto W_{(a, 0^n)}^{\otimes 2} = Z^a \otimes Z^a$ on two copies of the Hilbert space. Then, with the characters labelled by $b\in \mathbb{Z}_2^n$ and given by $\chi_b(a) = (-1)^{a\cdot b}$, the character POVM from \cref{eq:povm} becomes
\begin{equation}
    \Pi_b = \frac{1}{2^n}\sum_a (-1)^{a\cdot b}Z^a\otimes Z^a\,.
\end{equation}
This is precisely the computational difference sampling POVM given in Ref.\ \cite{grewalEfficientLearningQuantum2024}. This correspondence readily also generalizes to qudits.

\subsection{Common eigenbasis of the Weyl operators}\label{ssec:common-eigenbasis-of-weyls}
To implement the POVM in \cref{eq:povm-weyl}, we can measure in the joint eigenbasis of the mutually commuting $ W_x^{\otimes D}$ for all $x\in\mathbb{Z}_{d}^{2n} $.  But what does this joint eigenbasis look like and how exactly do we implement a measurement in it?

In the qubit case, the eigenbasis of all the $W_x^{\otimes 2}$ is the well-known multi-qubit Bell basis $\{\ket{W_x}\}_{x\in\mathbb{Z}_{2}^{2n}} $, with
\begin{equation}
    \ket{W_x} = (W_x \otimes I )\ket{\Omega}, \quad \ket{\Omega} = \frac{1}{\sqrt{2^n}}\sum_{x\in \mathbb{Z}_{2}^{n}}\ket{x}\ket{x}\,.
\end{equation}
Consequently, since $D=2d=4$, the eigenbasis of $W_x^{\otimes 4}$ consists of tensor products of Bell states $\ket{W_x}\otimes \ket{W_y}$ and the POVM in  \cref{eq:povm-weyl} corresponds to Bell difference sampling.

In the qudit case, the common eigenbasis of the $W_x^{\otimes D}$ is not given by the generalized Bell states $\ket{W_x} = (W_x \otimes I )\ket{\Omega}$ where $\ket{\Omega}$ is the qudit maximally entangled state.
In fact, these states live on 2 copies, however we are looking for a basis of the $d$-copy Hilbert space $(\CC^{d^n})^{\otimes{d}}$ since for odd $d$, we have $D=d$. This basis can be constructed as follows.

Since the $W_x^{\otimes d}$ mutually commute for all $x \in \mathbb{Z}_{d}^{2n}$, they form a stabilizer group $S_d$, namely
\begin{equation}
    S_d = \langle X^{\otimes{d}}, Z^{\otimes{d}}\rangle^{\otimes{n}}\,.
\end{equation}
In general, to any stabilizer group, one can associate a joint eigenbasis of stabilizer states, a so-called \textit{stabilizer basis} and a measurement in such a stabilizer basis can be implemented by application of a suitable Clifford circuit followed by measurement in the computational basis. Here, by the tensor product structure of $S_d$, the problem reduces to constructing a Clifford circuit on $d$ qudits that maps $X^{\otimes d}\mapsto Z_1, Z^{\otimes d} \mapsto Z_2$. In tableau form, this Clifford can be found via standard Gaussian elimination over $F_d$, and consequently be compiled into a Clifford circuit \cite{hostensStabilizerStatesClifford2005}.

More generally, the number of independent generators of $S_d$ is $2n$, hence for $d>2$, $S_d$ is not a maximal stabilizer group, as the maximum possible number of independent generators is $dn>2n$ on $dn$ qudits. Consequently, the stabilizer basis associated to $S_d$ is not uniquely determined. While the exact choice is not important here---any choice does the job---we emphasize again that the basis can be chosen to tensorize along the $n$ qudits. In other words, since $((\mathbb{C}^d)^{\otimes{n}})^{\otimes{d}} \cong ((\mathbb{C}^d)^{\otimes{d}})^{\otimes{n}}$, the stabilizer basis can be chosen to only act simultaneously on $d$ copies of the $i$-th qudit. Hence, assuming all-to-all connectivity, the measurement can be implemented by a Clifford circuit of depth at most $O(d)$. 

\subsection{Learning stabilizer codes from access to code states}\label{ssec:learning-code}
In quantum learning theory, the standard access model typically involves learning properties of a single quantum state. However, for the symmetry identification tasks explored in this work, it is both natural and beneficial to consider access to a collection of unknown quantum states that exhibit the same underlying symmetry. A canonical example of this arises with states residing in the code space of a common stabilizer code. It is entirely plausible to encounter scenarios where the goal is to identify the stabilizer code based on access to multiple different states in the code space, rather than repeated access to a single one. This situation might, for instance, occur when intercepting a stream of encoded quantum information being transmitted over a quantum communication channel.

In this vein, here we define the following variant of the hidden stabilizer group problem from \cref{def:hidden-stabilizer-group}:
\begin{defi}[Learning hidden stabilizer codes]\label{def:learning-code}
Let $\{\ket{\psi_i}\}_i\in\CC^{d^n}$ be a collection of states on $n$ qudits all belonging to the code space of some stabilizer code. In other words, assume that there is a stabilizer group $S\subseteq \mathcal{P}_n$, such that  for all $i$, 
\begin{equation}
P\ket{\psi_i} = \ket{\psi_i}, \quad \forall\, P \in S\,,
\end{equation}
while for all $\ket{\psi_i}$ and $P \in \mathrm{P}_n \setminus \mathrm{S}$, $|\bra{\psi_i} P \ket{\psi_i}| \leq 1-\epsilon$.
The hidden stabilizer code problem asks to identify $S$ given access
to copies of $\{\ket{\psi_i}\}_i$.
\end{defi}

We can solve this problem by essentially following the same steps as outlined in \cref{ssec:learning-weyl}. In particular, we again take the hidden subgroup $H$ to be the phaseless version of the stabilizer group $S$ which defines the code space.
We implement the same POVM from \cref{eq:povm-weyl} but on $\bigotimes_{i=1}^D \ket{\psi_i}\bra{\psi_i}$ in order to sample from the distribution
\begin{equation}\label{eq:q_multiple}
    q_{\psi_1,\ldots,\psi_D}(y)= \frac{1}{d^{2n}}\sum_{x\in \mathbb{Z}_{d}^{2n}}\om^{-[y,x]} \tr \left(  W_x^{\otimes D} \bigotimes_{i=1}^D \ket{\psi_i}\bra{\psi_i} \right)\,.
\end{equation}
Crucially, using the same analysis as outlined already in \cref{sec:our-approach-to-abelian-state-hsp}, we find that $q_{\psi_1,\ldots,\psi_D}$ is only supported on $H^\perp$ and not concentrated on any subgroup.

\subsection{Detecting global symmetries of symmetry protected topological ordered states}\label{ssec:learning-global-stabilizer-symmetries}

Beyond quantum‐information theory, our results also connect to problems in condensed‐matter physics—specifically, to detecting \emph{symmetry‑protected topological} (SPT) phases, which we capture as instances of the StateHSP. At the heart of SPT order is the principle that, under a given symmetry group $H$, distinct SPT states cannot be connected by any constant depth, symmetry‑preserving quantum circuit without undergoing a phase transition \cite{zengQuantumInformationMeets2019a}. In contrast, if symmetry is allowed to be broken, any SPT state can be smoothly deformed into a trivial product state without encountering a phase transition.

There are simple instances of stabilizer states known to exhibit SPT order. A canonical example is the \emph{cluster} state vector $\ket\psi$ on $n$ qubits on a ring, with $d=2$ and $n$ even. It is a qubit stabilizer state with 
the stabilizer group $S$ being generated by the $n$ stabilizers $K_j=Z_{j-1}\otimes X_j \otimes Z_{j+1}$ (with indices taken mod $n$) each acting on three consecutive qubits labeled by $j\in\{0,\dots, n-1\} $ \cite{elseSymmetryProtectionMeasurementbased2012}.
Then $|\psi\rangle$ is the unique stabilizer state 
vector satisfying $P|\psi\rangle = |\psi\rangle,\,  \forall P\in S$.  
It has the 
global symmetries
\begin{eqnarray}
	P_e&\coloneqq &\bigotimes_{j=0}^{n/2-1} (I\otimes X)^{\otimes n/2},\\
	P_o& \coloneqq &\bigotimes_{j=0}^{n/2-1} (X\otimes I)^{\otimes n/2},
\end{eqnarray}
with Pauli $X$ operators supported on even and odd sites, in that 
\begin{equation}
	P_e|\psi\rangle = |\psi\rangle,\,\, 
	P_o|\psi\rangle = |\psi\rangle
\end{equation}	
holds true. This symmetry is referred to as a $\ZZ_2 \times \ZZ_2$-symmetry.
Although $|\psi\rangle$ admits an explicit matrix product state representation---obtained by applying a layer of controlled‑$Z$ gates to a product state---any constant depth quantum circuit mapping it to the trivial product state vector $\ket{+}^{\otimes n}$ must break the $\ZZ_2\times\ZZ_2$ symmetry.  Consequently, the cluster state and the trivial product state, while sharing the global symmetry group generated by $P_e,P_o$, belong to distinct SPT phases \cite{zengQuantumInformationMeets2019a,ciracMatrixProductStates2021}.
Motivated by this, we define the hidden global symmetry problem, closely related to the
hidden stabilizer group problem from \cref{def:hidden-stabilizer-group}.

\begin{defi}[Hidden global symmetry problem] Let $p\in \NN$ be a period and let $\ket{\psi}\in \CC^{d^n}$ be a state vector on $n$ qudits on a ring, where $n/p\in \NN$. Let ${\cal M}_p <\mathcal{P}_p$ be an abelian subgroup of the Pauli group on $p$ sites and let $\mathrm{P}_p,\mathrm{M}_p $ denote the respective phaseless versions. Then, suppose that $\ket{\psi}$ is invariant under the stabilizer group of the form
\begin{equation}
	\{ P^{\otimes (n/p)}, \, P \in \mathcal{M}_p \},
\end{equation}
 while for all $P\in \mathrm{P}_p \setminus \mathrm{M}_p$, $|\bra{\psi} P^{\otimes (n/p)} \ket{\psi}| \leq 1-\epsilon$. 
The hidden global symmetry problem asks to identify this stabilizer group.
\end{defi}

This is an on-site symmetry upon ``blocking''. In the above example of the cluster state, the hidden global symmetry subgroup is $\langle P_e, P_o\rangle$. To view this as an abelian StateHSP, we observe that to learn the global symmetry, it suffices to learn the on-site symmetry group $\mathcal{M}_p$. Then, as in  \cref{ssec:learning-weyl}, we focus again on first learning it up to phases, i.e., we set $H$ to be the phaseless version of $\mathcal{M}_p$. Further, we choose as a parent group $G=\mathbb{Z}_d^{2p}$ with the action $x \in \mathbb{Z}_d^{2p} \mapsto R(x) = (W_x^{\otimes n/p})^{\otimes D}$. Finally, it follows from \cref{thm:sample-complexity} that this can be efficiently solved with a complexity independent of $n$, to capture the symmetry in a quantum phase of matter.

\begin{thm}[Sample complexity lower bound for the hidden global symmetry problem] There exists an algorithm for solving the 
 hidden global symmetry problem with sample complexity $O(p \log d
 /(d\epsilon))$.
\end{thm}

\section{Revisiting the hidden cut problem}
\label{sec:hidden-cut}
In this section, we are concerned with the following problem first studied in Ref.\ \cite{boulandStateHiddenSubgroup2025}:
\begin{defi}[Hidden many-cut problem]\label{defi:hidden-many-cut}
 Let $\ket{\psi}$ be a state vector on $n$ qubits. Suppose that $\ket{\psi}$
is a product of $m\geq1$ factor state vectors
\begin{equation} \label{eq:product-state}
\ket{\psi}=\ket{\phi_{1}}_{C_{1}}\otimes\cdots\otimes\ket{\phi_{m}}_{C_{m}}
\end{equation}
for some partition $C_{1}\sqcup\cdots\sqcup C_{m}=\left[n\right]$
such that each factor state vector $\ket{\phi_{k}}_{C_{k}}$ is at least
$\epsilon$-far in trace distance from any multipartite product state on $\left|C_{k}\right|$
qubits. The hidden many-cut problem asks to identify the set partition
$C_{1}\sqcup\cdots\sqcup C_{m}$, given copies of the state vector $\ket{\psi}$.
\end{defi}

As noted previously in Ref.\ \cite{boulandStateHiddenSubgroup2025}, this problem fits into the StateHSP framework by choosing $G =\mathbb{Z}_{2}^{n}$ and the representation
\begin{align}
    R: \mathbb{Z}_{2}^{n} &\to \mathrm{U}\left(((\mathbb{C}^2)^{\otimes n} )^{\otimes{2}} \right)\, , \\
    x &\mapsto R(x) = \bigotimes_{i=1}^n \mathrm{SWAP}^{x_i}\, .
\end{align}
We emphasize that $R$ is a representation on $\left((\mathbb{C}^2)^{\otimes n} \right)^{\otimes{2}}$, i.e., on two copies of the $n$-qubit Hilbert space.
The hidden subgroup $H$ is given by
\begin{equation}\label{eq:hidden-cut-subspace}
    H= \mathrm{span} \{1^n, 0^{C_1}1^{\overline{C_1}}, \dots, 0^{C_{m-1}}1^{\overline{C_{m-1}} } \}
\end{equation}
where $\overline{C_i}$ denotes the complement of $C_i$ in $[n]$, and $H \cong \mathbb{Z}_2^{m}$. It can be verified that if $\ket\psi$ is of the form given in \Cref{eq:product-state}, then $\ket{\psi}^{\otimes{2}} $ is invariant under $R(x)$ for all $x\in H$ as defined in \Cref{eq:hidden-cut-subspace}. 

\begin{prop}[c.f. Proposition 5 in Ref.\ \cite{boulandStateHiddenSubgroup2025}]
 \label{prop:non-cuts}  Let $\ket{\psi}$ be a state vector on $n$ qubits. Suppose that $\ket{\psi}$
is a product of $m\geq1$ factor state vectors
\begin{equation} 
\ket{\psi}=\ket{\phi_{1}}_{C_{1}}\otimes\cdots\otimes\ket{\phi_{m}}_{C_{m}}
\end{equation}
such that each factor state vector $\ket{\phi_{k}}_{C_{k}}$ is $\epsilon$-far in trace distance from any multipartite product state on $\left|C_{k}\right|$
qubits. Let $S\subseteq[n]$ be a subset
of qubits. Then, 
\begin{equation}
\tr \left(\psi_{S}^{2}\right)=\begin{cases}
1 & S= C_i \text{ for }i \in [m] \,,\\
\leq1-\epsilon^{2} & \mathrm{else} \,.
\end{cases}
\end{equation}
Here, $\tr \left(\psi_{S}^{2}\right)$ denotes the purity of the reduced state on $S$.
\end{prop}

Since $\mathbb{Z}_{2}^{n}$ is abelian, we can tackle this problem via Fourier sampling as outlined in \cref{sec:our-approach-to-abelian-state-hsp}. In particular, we will implement the character POVM corresponding to \Cref{eq:povm}. The characters are labeled by $y\in \mathbb{Z}_{2}^{n}$ and given by
\begin{equation}
    \chi_{y}(x) = (-1)^{y\cdot x}\,,
\end{equation}
for all $x\in \mathbb{Z}_{2}^{n}$ and the POVM is given by
\begin{equation}\label{eq:povm-swap}
    \Pi_y = \frac{1}{2^n} \sum_{x\in \mathbb{Z}_{2}^{n}} (-1)^{y\cdot x} \bigotimes_{i=1}^n \mathrm{SWAP}^{x_i}\,.
\end{equation}
Hence, the resulting output distribution $q_\psi(y) = \tr \big(\Pi_y\; (\ket \psi \bra \psi)^{\otimes 2} \big)$ is given by
\begin{equation}
    q_{\psi}(y)= \frac{1}{2^n} \sum_{x\in \mathbb{Z}_{2}^{n}}(-1)^{y\cdot x} \tr(\psi_x^2) \,. 
\end{equation}
Now, applying \Cref{thm:sample-complexity} and \Cref{prop:non-cuts}, we find that the hidden many-cut problem as stated in \cref{defi:hidden-many-cut} on $n$ qubits can be solved using $O(n/\epsilon^2)$ copies of the unknown state vector $\ket{\psi}$.

The projective measurement in \Cref{eq:povm-swap} can be implemented by measuring in the common eigenbasis of the $R(x) = \bigotimes_{i=1}^n \mathrm{SWAP}^{x_i}$ which is the multi-qubit Bell basis. Concretely, we divide the Bell state vectors into triplet and singlet as
\begin{align}
 \ket{0, 0} &=\frac{1}{\sqrt{2}} \left( |0,0\rangle + |1,1\rangle \right),  \\
 \ket{0, 1} &= \frac{1}{\sqrt{2}} \left( |0,0\rangle -|1,1\rangle \right) ,\\
 \ket{0, 2} &=  \frac{1}{\sqrt{2}} \left( |0,1\rangle + |1,0\rangle \right),\\
 \ket{1, 1} &= \frac{1}{\sqrt{2}} \left( |0,1\rangle - |1,0\rangle \right) .
\end{align}
To summarize, we arrive at the following theorem.

\begin{thm}[Improved hidden-cut algorithm]\label{thm:hidden-cut-thm}
There is an efficient non-adaptive quantum algorithm for the hidden
many-cut problem on $n$ qubits which uses $O\left(n/\epsilon^{2}\right)$
copies of the unknown state and runs in polynomial time, requiring circuits
of constant depths acting on two copies at a time using no 
additional auxiliary qubits.
\end{thm}

As mentioned before, via the \emph{Choi Jamiołkowski isomorphism}, the 
many-cut problem for quantum states naturally translates into a 
\emph{many-cut problem for unitaries}. This translation allows us to apply our method by preparing the Choi state of the unitary channel.
In this way, given oracle access to a unitary, one can efficiently learn in what hidden way the unitary is a product.

\section{Translation invariance as StateHSP}
Here we restrict ourselves to qubits for simplicity and consider 
the cyclic group $C_n=\langle g\rangle$ where $g$ is some canonical generator such that $g^n=1$. We can define a representation of this group 
acting on $n$ qubits on a ring as
\begin{equation}\label{eq:T}
T\ket{x_1,x_2,\ldots,x_n}=\ket{x_n,x_1,\ldots,x_{n-1}}
\end{equation}
where $T=T(g)$ is the canonical 
translation operator corresponding to $g$ and therefore $T(g^k)=T^k$. Equivalently, we can identify the cyclic group with $\ZZ_n=\{0,1,\ldots,n-1\}$ with addition modulo $n$.

\begin{defi}[Hidden translation problem]
    Let $\ket\psi\in\CC^n$ be a state vector of $n$ qubits. We assume there is a subgroup $H$ of $\ZZ_n$ such that
    \begin{equation}
        T^k\ket\psi=\ket\psi,\quad \forall k\in H
    \end{equation}
    and also that $|\bra\psi T^k\ket\psi|<1-\ep$ whenever $k\notin H$. The hidden translation problem asks to identify the subgroup $H$.
\end{defi}

\noindent This problem fits into the StateHSP framework by choosing $G=\ZZ_n$ and the representation
\begin{align}
    R:\ZZ_n &\to\mathrm{U}\left(\CC^n\right) \,, \\
    k &\mapsto R(k) = T^k \,,
\end{align}
where $T$ is the operator defined in \cref{eq:T}. Since $\ZZ_n$ 
is abelian, we can tackle this problem via our approach to the abelian StateHSP outlined in \cref{sec:our-approach-to-abelian-state-hsp}. In particular, we will implement the character POVM corresponding to \Cref{eq:povm}. The irreducible characters of $\ZZ_n$ can be labeled by $j\in\ZZ_n$ and are given by
\begin{equation}
    \chi_j(k)=\om^{jk}\,,
\end{equation} for all $k\in\ZZ_n$, where $\om=\e^{2\pi i/n}$ is a primitive $n$-th root of unity.
Therefore, in this case the POVM corresponding to \cref{eq:povm} reads
\begin{equation}\label{eq:povm-translation}
    \Pi_j = \frac{1}{n}\sum_{k\in\ZZ_n} \om^{-jk} T^k \,.
\end{equation}
Hence, the resulting output distribution $q_\psi(y) = \tr (\Pi_j\; \ket \psi \bra \psi )$ is given by
\begin{equation}
    q_{\psi}(j)= \frac{1}{n}\sum_{k\in\ZZ_n}\om^{-jk}\bra\psi T^k\ket\psi\,.
\end{equation}

Note that all subgroups of $\ZZ_n$ can be labelled by their period since all those subgroups are of the form $\{k\in\ZZ_n:k=0\mod r\}$
with $r|n$ (denoting that $r$ is a divisor of $n$). So, the hidden subgroup $H$ is of the form
\begin{equation}
    H_r=\{0,r,2r,\ldots,(n/r-1)r\}
\end{equation}
with $|H_r|=n/r$ for every divisor $r$.

In principle, the projective measurement $\Pi_j$ 
from \Cref{eq:povm-translation} can be implemented without 
any auxiliary 
qubits by measuring in the common eigenbasis of the $T^k$. 
However, 
from a practical point of view, it is not clear how to 
implement this eigenbasis measurement via a quantum circuit. Instead, 
for a concrete implementation of $\Pi_j$, one can employ the 
generalized phase estimation approach laid out in \cref{ssec:quantum-fourier-sampling}. Concretely, we will need an auxiliary register 
with $m = \left \lceil{\log_2(n)}\right \rceil$ many qubits 
which 
will serve as the control register for implementing the controlled group action $\sum_{k=0}^{m-1} \ket{k}\bra{k}\otimes T^k$. 
Assuming all-to-all connectivity, for all 
$k\in \{0,\dots, n-1 \}$, $T^k$ can be implemented via 
c-SWAPS in depth $O(n)$. Hence, to implement the whole 
controlled group action, via concatenating c-$T^{2^j}$ layers for $j=0,\dots, m-1$, we require a circuit of depth $O(n \log n)$. This circuit depth dominates the 
cost of subsequently implementing the QFT on the 
auxiliary register of size $O(\log n)$ qubits.
Hence, by virtue of our general 
result \cref{thm:sample-complexity}, we arrive at the following theorem.

\begin{thm}[Efficient algorithm for the hidden translation problem]\label{thm:hidden-translation}
There is an efficient non-adaptive quantum algorithm for the hidden
translation problem on $n$ qubits which uses $O\left((\log n)/\epsilon \right)$
copies of the unknown state and runs in polynomial time, 
requiring circuits
of depths $O(n \log n)$ acting coherently on only a 
single copy at a time using $O(\log n)$ 
additional auxiliary systems. 
\end{thm}

\section{Conclusions and outlook}

In this work, we have identified a unified, efficient approach for solving the abelian state hidden subgroup problem that 
not only subsumes known quantum‐learning subroutines (such as Bell difference sampling for qubit stabilizer learning) but also extends naturally to new settings---including qudit stabilizer groups, the hidden many‐cut problem, and translational symmetries on lattice systems. This approach is conceptually streamlined and features significantly reduced circuit complexity compared to prior work, making it even more feasible for practical implementation. It also constitutes a new bridge between sampling problems with the mere use of showing a quantum advantage \cite{SupremacyReview} with proper quantum algorithms with practical applications. As such, it may bring quantum algorithms featuring some utility and showing a quantum advantage closer to feasibility \cite{MindTheGaps,VastWorld}.

Looking ahead, one natural direction is to generalize the StateHSP framework to continuous symmetry groups, thereby unlocking further applications across the board. It would also be interesting to identify further applications in the quantum many-body context. More broadly, we hope this work invites further studies that bring together ideas of quantum algorithms with those of quantum learning theory.
In fact, this intersection seems to be a particularly fruitful avenue for future research.

 \section*{Acknowledgements}

We would like to thank Marios Ioannou for helpful discussions. This project has been supported by the German Federal Ministry for Research, Technology, and Space (MUNIQC-Atoms, QSolid, QuSol, Hybrid++), Berlin Quantum, the Munich Quantum Valley (K-8), the QuantERA (HQCC), the Clusters of Excellence (MATH+, ML4Q), the Quantum Flagship (Millenion, PasQuans2), 
the DFG (CRC 183, SPP 2541), and the European Research Council (DebuQC).

\bibliographystyle{alphaurl}

\bibliography{refs}

\newcommand{\etalchar}[1]{$^{#1}$}
\begin{thebibliography}{BGTW25}

\bibitem[AA24]{anshuSurveyComplexityLearning2024}
Anurag Anshu and Srinivasan Arunachalam.
\newblock A survey on the complexity of learning quantum states.
\newblock {\em Nature Reviews Physics}, 6:59--69, 2024.
\newblock \href {https://doi.org/10.1038/s42254-023-00662-4} {\path{doi:10.1038/s42254-023-00662-4}}.

\bibitem[ADIS24]{allcockBellSamplingStabilizer2024}
Jonathan Allcock, Joao~F. Doriguello, G{\'a}bor Ivanyos, and Miklos Santha.
\newblock Beyond {{Bell}} sampling: Stabilizer state learning and quantum pseudorandomness lower bounds on qudits, 2024.
\newblock \href {https://arxiv.org/abs/2405.06357} {\path{arXiv:2405.06357}}.

\bibitem[AdW17]{arunachalamGuestColumnSurvey2017b}
Srinivasan Arunachalam and Ronald de~Wolf.
\newblock Guest {{column}}: {{A survey}} of {{quantum learning theory}}.
\newblock {\em SIGACT News}, 48:41--67, 2017.
\newblock \href {https://doi.org/10.1145/3106700.3106710} {\path{doi:10.1145/3106700.3106710}}.

\bibitem[AG08]{aaronsonIdentifyingStabilizerStates2008}
Scott Aaronson and Daniel Gottesman.
\newblock Identifying {{stabilizer states}}.
\newblock 2008.
\newblock \href {https://doi.org/10.48660/08080052} {\path{doi:10.48660/08080052}}.

\bibitem[BFJ{\etalchar{+}}25]{bastidasUnificationFiniteSymmetries2025}
Victor~M. Bastidas, Nathan Fitzpatrick, K.~J. Joven, Zane~M. Rossi, Shariful Islam, Troy Van~Voorhis, Isaac~L. Chuang, and Yuan Liu.
\newblock Unification of finite symmetries in the simulation of many-body systems on quantum computers.
\newblock {\em Physical Review A}, 111:052433, 2025.
\newblock \href {https://doi.org/10.1103/PhysRevA.111.052433} {\path{doi:10.1103/PhysRevA.111.052433}}.

\bibitem[BGTW25]{boulandStateHiddenSubgroup2025}
Adam Bouland, Tudor Giurgic{\u a}-Tiron, and John Wright.
\newblock The {{state hidden subgroup problem}} and an {{efficient algorithm}} for {{locating unentanglement}}.
\newblock In {\em Proceedings of the 57th {{Annual ACM Symposium}} on {{Theory}} of {{Computing}}}, {{STOC}} '25, pages 463--470, New York, NY, USA, 2025. Association for Computing Machinery.
\newblock \href {https://doi.org/10.1145/3717823.3718118} {\path{doi:10.1145/3717823.3718118}}.

\bibitem[CBB{\etalchar{+}}23]{caiQuantumErrorMitigation2023}
Zhenyu Cai, Ryan Babbush, Simon~C. Benjamin, Suguru Endo, William~J. Huggins, Ying Li, Jarrod~R. McClean, and Thomas~E. O'Brien.
\newblock Quantum error mitigation.
\newblock {\em Reviews of Modern Physics}, 95:045005, 2023.
\newblock \href {https://doi.org/10.1103/RevModPhys.95.045005} {\path{doi:10.1103/RevModPhys.95.045005}}.

\bibitem[CGYZ25]{chenStabilizerBootstrappingRecipe2025}
Sitan Chen, Weiyuan Gong, Qi~Ye, and Zhihan Zhang.
\newblock Stabilizer {{bootstrapping}}: {{A recipe}} for {{efficient agnostic tomography}} and {{magic estimation}}.
\newblock In {\em Proceedings of the 57th {{Annual ACM Symposium}} on {{Theory}} of {{Computing}}}, {{STOC}} '25, pages 429--438, New York, NY, USA, 2025. Association for Computing Machinery.
\newblock \href {https://doi.org/10.1145/3717823.3718191} {\path{doi:10.1145/3717823.3718191}}.

\bibitem[Chi25]{childsLectureNotesQuantum2025}
Andrew~M. Childs.
\newblock Lecture {{notes}} on {{quantum algorithms}}.
\newblock 2025.
\newblock URL: \url{www.cs.umd.edu/~amchilds/qa/qa.pdf}.

\bibitem[CHW07]{childsWeakFourierSchurSampling2007}
Andrew~M. Childs, Aram~W. Harrow, and Pawe{\l} Wocjan.
\newblock Weak {{Fourier-Schur}} sampling, the hidden subgroup problem, and the quantum collision problem.
\newblock In {\em Proceedings of the 24th Annual Conference on {{Theoretical}} Aspects of Computer Science}, {{STACS}}'07, pages 598--609, Berlin, Heidelberg, 2007. Springer-Verlag.

\bibitem[CLL24]{chiaEfficientLearning$t$doped2024}
Nai-Hui Chia, Ching-Yi Lai, and Han-Hsuan Lin.
\newblock Efficient learning of $t$-doped stabilizer states with single-copy measurements.
\newblock {\em Quantum}, 8:1250, 2024.
\newblock \href {https://doi.org/10.22331/q-2024-02-12-1250} {\path{doi:10.22331/q-2024-02-12-1250}}.

\bibitem[CPSV21]{ciracMatrixProductStates2021}
J.~Ignacio Cirac, David {P{\'e}rez-Garc{\'i}a}, Norbert Schuch, and Frank Verstraete.
\newblock Matrix product states and projected entangled pair states: {{Concepts}}, symmetries, theorems.
\newblock {\em Reviews of Modern Physics}, 93:045003, 2021.
\newblock \href {https://doi.org/10.1103/RevModPhys.93.045003} {\path{doi:10.1103/RevModPhys.93.045003}}.

\bibitem[Cv10]{childsQuantumAlgorithmsAlgebraic2010}
Andrew~M. Childs and Wim {van Dam}.
\newblock Quantum algorithms for algebraic problems.
\newblock {\em Reviews of Modern Physics}, 82:1--52, 2010.
\newblock \href {https://doi.org/10.1103/RevModPhys.82.1} {\path{doi:10.1103/RevModPhys.82.1}}.

\bibitem[EBD12]{elseSymmetryProtectionMeasurementbased2012}
Dominic~V. Else, Stephen~D. Bartlett, and Andrew~C. Doherty.
\newblock Symmetry protection of measurement-based quantum computation in ground states.
\newblock {\em New Journal of Physics}, 14:113016, 2012.
\newblock \href {https://doi.org/10.1088/1367-2630/14/11/113016} {\path{doi:10.1088/1367-2630/14/11/113016}}.

\bibitem[EHW{\etalchar{+}}20]{eisertQuantumCertificationBenchmarking2020a}
Jens Eisert, Dominik Hangleiter, Nathan Walk, Ingo Roth, Damian Markham, Rhea Parekh, Ulysse Chabaud, and Elham Kashefi.
\newblock Quantum certification and benchmarking.
\newblock {\em Nature Reviews Physics}, 2:382--390, 2020.
\newblock \href {https://doi.org/10.1038/s42254-020-0186-4} {\path{doi:10.1038/s42254-020-0186-4}}.

\bibitem[EP25]{MindTheGaps}
Jens Eisert and John Preskill.
\newblock Mind the gaps: The fraught road to quantum advantage.
\newblock 2025.
\newblock \href {https://arxiv.org/abs/2510.19928} {\path{arXiv:2510.19928}}.

\bibitem[GIKL24a]{grewalEfficientLearningQuantum2024}
Sabee Grewal, Vishnu Iyer, William Kretschmer, and Daniel Liang.
\newblock Efficient {{learning}} of {{quantum states prepared with few non-Clifford gates}}, 2024.
\newblock \href {https://arxiv.org/abs/2305.13409} {\path{arXiv:2305.13409}}.

\bibitem[GIKL24b]{grewalImprovedStabilizerEstimation2024d}
Sabee Grewal, Vishnu Iyer, William Kretschmer, and Daniel Liang.
\newblock Improved {{stabilizer estimation}} via {{Bell difference sampling}}.
\newblock In {\em Proceedings of the 56th {{Annual ACM Symposium}} on {{Theory}} of {{Computing}}}, {{STOC}} 2024, pages 1352--1363, New York, NY, USA, 2024. Association for Computing Machinery.
\newblock \href {https://doi.org/10.1145/3618260.3649738} {\path{doi:10.1145/3618260.3649738}}.

\bibitem[GNW21]{grossSchurWeylDuality2021}
David Gross, Sepehr Nezami, and Michael Walter.
\newblock Schur--{{Weyl duality}} for the {{Clifford group}} with {{applications}}: {{Property testing}}, a {{robust Hudson theorem}}, and de {{Finetti representations}}.
\newblock {\em Communications in Mathematical Physics}, 385:1325--1393, 2021.
\newblock \href {https://doi.org/10.1007/s00220-021-04118-7} {\path{doi:10.1007/s00220-021-04118-7}}.

\bibitem[Har05]{harrowApplicationsCoherentClassical2005}
Aram~W. Harrow.
\newblock {\em Applications of Coherent Classical Communication and the {{Schur}} Transform to Quantum Information Theory}.
\newblock PhD thesis, arXiv, 2005.
\newblock \href {https://arxiv.org/abs/quant-ph/0512255} {\path{arXiv:quant-ph/0512255}}.

\bibitem[HCMP25]{VastWorld}
Hsin-Yuan Huang, Soonwon Choi, Jarrod~R. McClean, and John Preskill.
\newblock The vast world of quantum advantage.
\newblock 2025.
\newblock \href {https://arxiv.org/abs/2508.05720} {\path{arXiv:2508.05720}}.

\bibitem[HDM05]{hostensStabilizerStatesClifford2005}
Erik Hostens, Jeroen Dehaene, and Bart~De Moor.
\newblock Stabilizer states and {{Clifford}} operations for systems of arbitrary dimensions, and modular arithmetic.
\newblock {\em Physical Review A}, 71:042315, 2005.
\newblock \href {https://doi.org/10.1103/PhysRevA.71.042315} {\path{doi:10.1103/PhysRevA.71.042315}}.

\bibitem[HE23]{SupremacyReview}
Dominik Hangleiter and Jens Eisert.
\newblock Computational advantage of quantum random sampling.
\newblock {\em Rev. Mod. Phys.}, 95:035001, 2023.
\newblock \href {https://doi.org/10.1103/RevModPhys.95.035001} {\path{doi:10.1103/RevModPhys.95.035001}}.

\bibitem[HG24]{hangleiterBellSamplingQuantum2024}
Dominik Hangleiter and Michael~J. Gullans.
\newblock Bell sampling from quantum circuits.
\newblock {\em Physical Review Letters}, 133:020601, 2024.
\newblock \href {https://doi.org/10.1103/PhysRevLett.133.020601} {\path{doi:10.1103/PhysRevLett.133.020601}}.

\bibitem[HH25]{hinscheSingleCopyStabilizerTesting2025}
Marcel Hinsche and Jonas Helsen.
\newblock Single-{{copy stabilizer testing}}.
\newblock In {\em Proceedings of the 57th {{Annual ACM Symposium}} on {{Theory}} of {{Computing}}}, {{STOC}} '25, pages 439--450, New York, NY, USA, 2025. Association for Computing Machinery.
\newblock \href {https://doi.org/10.1145/3717823.3718169} {\path{doi:10.1145/3717823.3718169}}.

\bibitem[HKP21]{huangInformationTheoreticBoundsQuantum2021c}
Hsin-Yuan Huang, Richard Kueng, and John Preskill.
\newblock Information-{{theoretic bounds}} on {{quantum advantage}} in {{machine learning}}.
\newblock {\em Physical Review Letters}, 126:190505, 2021.
\newblock \href {https://doi.org/10.1103/PhysRevLett.126.190505} {\path{doi:10.1103/PhysRevLett.126.190505}}.

\bibitem[KR21]{klieschTheoryQuantumSystem2021a}
Martin Kliesch and Ingo Roth.
\newblock Theory of {{quantum system certification}}.
\newblock {\em PRX Quantum}, 2:010201, 2021.
\newblock \href {https://doi.org/10.1103/PRXQuantum.2.010201} {\path{doi:10.1103/PRXQuantum.2.010201}}.

\bibitem[LOH24]{leoneLearningTdopedStabilizer2024d}
Lorenzo Leone, Salvatore F.~E. Oliviero, and Alioscia Hamma.
\newblock Learning t-doped stabilizer states.
\newblock {\em Quantum}, 8:1361, 2024.
\newblock \href {https://doi.org/10.22331/q-2024-05-27-1361} {\path{doi:10.22331/q-2024-05-27-1361}}.

\bibitem[LRW23]{labordeTestingSymmetryQuantum2023}
Margarite~L. LaBorde, Soorya Rethinasamy, and Mark~M. Wilde.
\newblock Testing symmetry on quantum computers.
\newblock {\em Quantum}, 7:1120, 2023.
\newblock \href {https://doi.org/10.22331/q-2023-09-25-1120} {\path{doi:10.22331/q-2023-09-25-1120}}.

\bibitem[MdW16]{montanaroSurveyQuantumProperty2016c}
Ashley Montanaro and Ronald de~Wolf.
\newblock {\em A {{Survey}} of {{quantum property testing}}}.
\newblock Graduate {{Surveys}}. Theory of Computing Library, 2016.
\newblock \href {https://doi.org/10.4086/toc.gs.2016.007} {\path{doi:10.4086/toc.gs.2016.007}}.

\bibitem[Mon17]{montanaroLearningStabilizerStates2017}
Ashley Montanaro.
\newblock Learning stabilizer states by {{Bell}} sampling, 2017.
\newblock \href {https://arxiv.org/abs/1707.04012} {\path{arXiv:1707.04012}}.

\bibitem[RLW25]{rethinasamyQuantumComputationalComplexity2025}
Soorya Rethinasamy, Margarite~L. LaBorde, and Mark~M. Wilde.
\newblock Quantum computational complexity and symmetry.
\newblock {\em Canadian Journal of Physics}, 103:215--239, 2025.
\newblock \href {https://doi.org/10.1139/cjp-2023-0260} {\path{doi:10.1139/cjp-2023-0260}}.

\bibitem[Sho94]{shorAlgorithmsQuantumComputation1994a}
Peter~W. Shor.
\newblock Algorithms for quantum computation: Discrete logarithms and factoring.
\newblock In {\em Proceedings 35th {{Annual Symposium}} on {{Foundations}} of {{Computer Science}}}, pages 124--134, 1994.
\newblock \href {https://doi.org/10.1109/SFCS.1994.365700} {\path{doi:10.1109/SFCS.1994.365700}}.

\bibitem[Ter15]{terhalQuantumErrorCorrection2015a}
Barbara~M. Terhal.
\newblock Quantum error correction for quantum memories.
\newblock {\em Reviews of Modern Physics}, 87:307--346, 2015.
\newblock \href {https://doi.org/10.1103/RevModPhys.87.307} {\path{doi:10.1103/RevModPhys.87.307}}.

\bibitem[WS24]{wakehamInferenceInterferenceInvariance2024}
David Wakeham and Maria Schuld.
\newblock Inference, interference and invariance: {{How}} the {{quantum Fourier transform}} can help to learn from data, 2024.
\newblock \href {https://arxiv.org/abs/2409.00172} {\path{arXiv:2409.00172}}.

\bibitem[ZCZW19]{zengQuantumInformationMeets2019a}
Bei Zeng, Xie Chen, Duan-Lu Zhou, and Xiao-Gang Wen.
\newblock {\em Quantum {{information meets quantum matter}}: {{From quantum entanglement}} to {{topological phases}} of {{many-body systems}}}.
\newblock Springer New York, 2019.

\end{thebibliography}

\end{document}